\documentclass[acmsmall,nonacm]{acmart}
\usepackage{booktabs}
\usepackage{diagbox}
\usepackage{enumitem}
\usepackage[ruled,vlined,linesnumbered]{algorithm2e}
\usepackage{wrapfig}
\usepackage{multirow}
\usepackage{adjustbox}
\usepackage{amsmath}
\usepackage{amsthm} 
\usepackage{color} 
\usepackage[dvipsnames]{xcolor}
\usepackage{tikz}
\usepackage{bm} 
\usepackage[normalem]{ulem}
\newcommand\redsout{\bgroup\markoverwith{\textcolor{red}{\rule[0.5ex]{2pt}{1pt}}}\ULon}

\newcommand{\extraspacing}{\vspace{3mm}\noindent}

\newcommand{\vgap}{\vspace{1mm}}

\theoremstyle{plain}
\newtheorem{theorem}{Theorem}[section]

\newtheorem{lemma}[theorem]{Lemma}

\newtheorem{corollary}[theorem]{Corollary}

\theoremstyle{definition}

\newtheorem{problem}[theorem]{Problem}
\newtheorem{example}[theorem]{Example}

\theoremstyle{remark}

\newcommand{\boxminipg}[2]{\begin{center}\fbox{\begin{minipage}{#1}#2\end{minipage}}\end{center}}

\newcommand{\myeqn}[1]{\begin{align}#1\end{align}}

\newcommand{\set}[1]{\left\{#1\right\}}

\newcommand{\restr}[2]{{\left.\kern-\nulldelimiterspace #1 \vphantom{|} \right|_{#2}}}

\newcommand{\geo}[1]{\textsf{Geometric}(#1)}
\newcommand{\tgeo}[2]{\textsf{TruncatedGeometric}(#1,#2)}
\newcommand{\unif}[2]{\textsf{Uniform}(#1,#2)}

\def\B{\mathcal{B}}
\def\C{\mathcal{C}}
\def\D{\mathcal{D}}

\def\P{\mathcal{P}}

\def\T{\mathcal{T}}

\def\bj{\bm{j}}

\def\Pr{\mathbf{Pr}}
\DeclareMathOperator*\expt{\mathbf{E}}
\newcommand{\cut}[1]{{}}
\newcommand{\calP}{\mathcal{P}}

\def\fp{\mathrm{p}}
\def\fq{\mathrm{q}}

\renewcommand{\max}[0]{\operatorname{max}}
\renewcommand{\min}[0]{\operatorname{min}}
\newcommand{\attset}{\operatorname{attset}}
\newcommand{\schema}{\operatorname{schema}}
\newcommand{\DA}{\textsf{DirectAccess}~}
\def\rdX{\mathbf{X}}

\def\rdI{\mathbf{I}}

\newcommand{\bu}{\bm{u}}
\newcommand{\bv}{\bm{v}}
\newcommand{\att}{\ensuremath{\mathbf{att}}}
\newcommand{\dom}{\textbf{dom}}
\newcommand{\In}{\ensuremath{N}}
\newcommand{\join}{\ensuremath{\operatorname{Join}}}

\newcommand{\key}{\textsf{key}}
\newcommand{\prob}{\gamma}

\begin{document}

\title{Subset Sampling over Joins}

\author{Aryan Esmailpour}
\orcid{0009-0000-3798-9578}
\affiliation{%
  \institution{
  University of Illinois Chicago}
  \city{Chicago}
  \country{USA}
}
\email{aesmai2@uic.edu}

\author{Xiao Hu}
\orcid{0000-0002-7890-665X}
\affiliation{%
  \institution{University of Waterloo}
  \city{Waterloo}
  \country{Canada}
}
\email{xiaohu@uwaterloo.ca}

\author{Jinchao Huang}
\orcid{0009-0009-2902-259X}
\affiliation{%
  \institution{The Chinese University of Hong Kong}
  \city{Hong Kong}
  \country{Hong Kong}
}
\email{jchuang@link.cuhk.edu.hk}

\author{Stavros Sintos}
\orcid{0000-0002-2114-8886}
\affiliation{%
  \institution{
  University of Illinois Chicago}
  \city{Chicago}
  \country{USA}
}
\email{stavros@uic.edu}

\begin{abstract}
Subset sampling (also known as Poisson sampling), where the decision to include any specific element in the sample is made independently of all others, is a fundamental primitive in data analytics, enabling efficient approximation by processing representative subsets rather than massive datasets.
While sampling from explicit lists is well-understood, modern applications---such as machine learning over relational data---often require sampling from a set defined implicitly by a relational join. In this paper, we study the problem of \emph{subset sampling over joins}: drawing a random subset from the join results, where each join result is included independently with some probability. We address the general setting where the probability is derived from input tuple weights via decomposable functions (e.g., product, sum, min, max). Since the join size can be exponentially larger than the input, the naive approach of materializing all join results to perform subset sampling is computationally infeasible. We propose the first efficient algorithms for subset sampling over acyclic joins: (1) a \emph{static index} for generating multiple (independent) subset samples over joins; (2) a \emph{one-shot} algorithm for generating a single subset sample over joins; (3) a \emph{dynamic index} that can support tuple insertions, while maintaining a one-shot sample or generating multiple (independent) samples.
Our techniques achieve near-optimal time and space complexity with respect to the input size and the expected sample size.
\end{abstract}

\maketitle

\section{Introduction}
\label{sec:intro}
Sampling is a cornerstone technique in modern data analytics and machine learning, allowing systems to trade a small amount of accuracy for significant gains in performance by processing a representative subset of data rather than the entire dataset~\cite{olken1993random, cormode2012synopses}. A particularly powerful primitive is \emph{subset sampling} (also known as \emph{Poisson sampling}), where the decision to include any specific element in the sample is made independently of all others, based on a specific probability or weight. This independence is often crucial for theoretical guarantees in randomized algorithms and allows for uncoordinated execution in distributed systems~\cite{motwani1995randomized, duffield2007priority}.

While subset sampling from explicit lists of items is well-understood, a new challenge arises in modern data pipelines where the dataset of interest is not stored in a single table but is defined \emph{implicitly} as the result of a relational join. In this paper, we study the problem of Subset Sampling over Joins: efficiently drawing a random subset from the join results, where every join result is associated with a weight and is included in the sample independently with probability equal to its weight. Crucially, we focus on the standard setting where the weight of a join result is not arbitrary, but is derived from weights assigned to the input tuples via an aggregation function---such as the product, minimum, maximum, or summation of the component weights. A compelling motivation for this problem lies in dataset condensation (or coreset construction) for machine learning over multi-relational data~\cite{kumar2015learning, schleich2016learning}. 

\begin{example}
    \label{example:intro} 
    Consider a high-stakes scientific computing task---virtual screening in drug discovery---where a researcher wishes to train a graph neural network to predict the efficacy of complex molecular interactions. The training data is defined by joining massive tables of {\em chemical compounds}, {\em protein targets}, and {\em metabolic pathways}. The full set of valid interaction subgraphs (the join result) is combinatorially explosive---often exceeding petabytes in size---making it impossible to materialize or exhaustively simulate. To make training feasible, the system must generate a condensed, representative dataset~\cite{feldman2011unified, bachem2017practical}. Here, input tuples carry crucial domain-specific weights: a compound has a binding affinity score, and a protein target has a clinical relevance score. The importance of a resulting interaction tuple is naturally derived from these components (e.g., the product of affinity and relevance, representing the joint probability of success; or the minimum, representing a bottleneck constraint). By performing subset sampling with these derived weights, we create a manageable, high-quality summary of the chemical space.
\end{example}

However, performing this sampling efficiently is algorithmically difficult due to the implicit nature of the data. A naive approach would be to first compute the join results as well as their aggregated weights (e.g., the product of its components), and then apply the classic subset sampling algorithm. This strategy is prohibitively costly because of the ``materialization bottleneck'': the size of the join result can be exponentially larger than the input relations \cite{atserias2013size}, which is very costly to compute and store. In this naive process, the algorithm would spend the vast majority of its time generating join results that are immediately discarded by the coin flips, wasting massive amounts of computation and memory. To the best of our knowledge, no prior work has addressed the problem of relational subset sampling except this baseline. In this paper, we aim to fill this gap by providing the first efficient subset sampling algorithms that operate directly on the input relations. By exploiting the join structure of the query and the decomposable nature of the weight functions, we can effectively ``skip'' over the vast sea of rejected tuples and directly generate the successful samples. Our collective results establish subset sampling over joins as a tractable database primitive, which will pave the way for scalable, interactive analytics over massive multi-relational datasets.

\subsection{Problem Definition}
\label{sec:prob-def}
{\bf Subset Sampling.} Given a set of elements $ S=\{e_1, e_2, \ldots, e_n\}$ and a function $\fp$ that assigns each element $e\in S$ a probability $\fp(e)$, a \emph{subset sampling} query asks for a \emph{subset sample} of $S$. A subset sample of $S$ is a random subset $\mathbf{X}\in 2^S$, where each element $e\in S$ is sampled into $\mathbf{X}$ independently with probability $\fp(e)$.
More formally, the distribution of $\mathbf{X}$ is:
        $
            \Pr[\mathbf{X}=Y]=\left( \prod_{e\in Y}\fp(e) \right)\left( \prod_{e\in { S\setminus Y}}(1-\fp(e)) \right), Y\subseteq  S.
        $
We refer to $\Psi=\left< S,\fp\right>$ as a {\em subset sampling problem instance}.
We denote the expected size of $\mathbf{X}$, i.e., $E[|\mathbf{X}|]=\sum_{e\in S}\fp(e)$, as $\mu_\Psi$.

\vgap

\noindent {\bf(Natural) Joins.} Let $\att$ be a set where each element is called an {\em attribute}, and $\dom$ be another set where each element is called a {\em value}. A {\em tuple} over a set $U \subseteq \att$ is a function $\bu: U \rightarrow \dom$. For any subset $U' \subseteq U$, define $\bu[U']$ as the tuple $\bu'$ over $U'$ such that $\bu'(A) = \bu(A)$ for every $A \in U'$. We also call $U$ the support attributes of $\bu$. A {\em relation} is a set $R$ of tuples over the same set $U$ of attributes; we call $U$ the {\em schema} of $R$, a fact denoted as $\schema(R) = U$. Given a subset $U \subseteq \schema(R)$, the {\em projection of $R$ on $U$}---denoted as $R[U]$---is a relation with schema $U$ defined as $R[U] = \set{\text{tuple $\bu$ over $U \mid \exists$ tuple $\bm{v} \in R$ s.t.\ $\bu[U] = \bm{v}[U]$}}.$
We represent a {\em natural join} (henceforth, simply a ``join'') as a set $Q$ of relations. We denote the number of relations in $Q$ as $k=|Q|$. The input size of $Q$ is defined as the total number of tuples in $Q$, denoted as $\In = \sum_{R \in Q}|R|$. Define $\attset(Q) = \bigcup_{R \in Q} \schema(R)$. The join result is the following relation over $\attset(Q)$:
\[
    \join(Q) = \left\{
    \textrm{tuple $\bu$ over $\attset(Q)$} \mid \forall R \in Q, \, \textrm{$\bu[\schema(R)] \in R$}
    \right\}.
\]
The join $Q$ can be characterized by a {\em schema graph} $G_Q = (V, E)$, which is a hypergraph where each vertex in $V$ is a distinct attribute in $\attset(Q)$, and each edge in $E$ is the schema of a distinct relation in $Q$. The set $E$ may contain identical edges because two (or more) relations in $Q$ can have the same schema. The term ``hyper'' suggests that an edge can have more than two attributes. The join $Q$ is {\em acyclic} \cite{yannakakis1981algorithms} if $G_Q$ is acyclic. A hypergraph is acyclic if there exists a tree $\mathcal{T}$, called a {\em join tree}, whose nodes have a one-to-one correspondence with the relations in $Q$ (i.e., hyperedges in $E$), such that for any attribute $A \in V$, the set of nodes corresponding to relations containing $A$ forms a connected subtree in $\mathcal{T}$. A join that is not acyclic is called a {\em cyclic join}.

As a variant of join, the semi-join between $R_1$ and $R_2$ is defined as $R_1 \ltimes R_2 = R_1 \Join (R_2[\schema(R_1) \cap \schema(R_2)])$, which returns all tuples in $R_1$ that can be joined with some tuple from $R_2$. When the second operand is a single tuple $\bv$, $R_1 \ltimes \bv$ is a shorthand for $R_1 \ltimes \{\bv\}$. As a special case, 
$R \ltimes\emptyset=R$.

\smallskip
\noindent {\bf Subset Sampling over Joins.} Let $Q = \{R_1, R_2, \ldots, R_k\}$ be a join as defined above. Each relation $R_j \in Q$ is associated with a function $\fp_j: R_j \rightarrow [0,1]$ that assigns a weight from $[0,1]$ to each tuple in it. 
For each join result $\bm{u} \in \join(Q)$, its weight is a function $\mathcal{F}$ defined on the weights of these tuples that constitute this join result: 
\myeqn{
\label{eq:join-probability}
    \fp(\bm{u}) = \mathcal{F}\left(\fp_1(\bm{u}[\schema(R_1)]), \fp_2(\bm{u}[\schema(R_2)]), \ldots, \fp_k(\bm{u}[\schema(R_k)])\right),
}
where $\mathcal{F}$ can be the \textsf{MAX}, \textsf{MIN}, \textsf{PRODUCT}, or \textsf{SUM} function. 
It is required that $\fp(\bm{u}) \in [0,1]$.
Given a join instance $Q$ and a set of associated weight functions $\{\fp_j\}_{j \in \{1, 2, \ldots, k\}}$, a {\em relational subset sampling} query is a subset sampling query to the subset sampling instance $\Psi=\left<\join(Q),\fp\right>$, i.e., asks to sample a random subset $\mathbf{X}$ of $\join(Q)$, where each join result $\bu$ is sampled into $\mathbf{X}$ independently with probability defined by its weight $\fp(\bu)$. More formally, the distribution of $\mathbf{X}$ is
\myeqn{
            \Pr[\mathbf{X} =Y] =\left(\prod_{\bu\in Y}\fp(\bu) \right)\left(\prod_{\bu\in {\join(Q) \setminus Y}}(1-\fp(\bu)) \right), Y \subseteq \join(Q).
}
We denote the expected size of $\mathbf{X}$, i.e., $E[|\mathbf{X}|] = \sum_{\bu \in \join(Q)} \fp(\bu)$, as $\mu_{\Psi}$.

The first problem we study is an {\em indexing} (data structure) problem, where we wish to have an index that supports drawing multiple independent subset samples from join results efficiently:
\boxminipg{0.9\linewidth}{
\begin{problem}[Indexed Subset Sampling over Joins]
\label{prob:static-index}
    Given a join instance $Q$ and a set of associated weight functions $\{\fp_i\}_{i\in\{1,2,\ldots,k\}}$, the goal is to construct an index for answering subset sampling queries to the instance $\left<\join(Q),\fp\right>$. Additionally, the random subsets returned by distinct queries must be independent. 
\end{problem}
}
We are interested in the preprocessing time $t_p$ and space usage of the index constructed, as well as the time $t_s$ for answering a relational subset sampling query from the index. 

The second problem we study is a {\em one-shot} version of relational subset sampling where we wish to output one subset sample from join results efficiently:
\boxminipg{0.9\linewidth}{
\begin{problem}[One-shot Subset Sampling over Joins]
\label{prob:static-one-shot}
    Given a join instance $Q$ and a set of associated weight functions $\{\fp_i\}_{i\in\{1,2,\ldots,k\}}$, the goal is to compute one subset sample from the instance $\left<\join(Q),\fp\right>$.
\end{problem}
}
In this case, we are interested in the total answering time. Although any solution to Problem~\ref{prob:static-index} gives a solution to Problem~\ref{prob:static-one-shot} with total time $O(t_p + t_s)$, this may not be efficient enough.

\smallskip\noindent {\bf Dynamic Subset Sampling over Joins.} We also consider the dynamic setting with tuple insertions\footnote{In the fully dynamic case with both insertions and deletions, the subset sampling problem (with all weights as $1$) is at least as hard as maintaining the join results. For general (more precisely, a non-hierarchical) join query, the update time must be $\Omega(\sqrt{N})$ just to maintain the Boolean answer, when both insertions and deletions are allowed \cite{berkholz17:_answer}. Even for this Boolean problem, establishing a matching upper bound is still open, except for some specific joins.}. Now, each insertion is a triple $\langle \bu, R_h, \eta, p \rangle$ for $\eta \in \mathbb{Z}^+$, $p\in [0,1]$ and $R_h \in Q$, indicating that tuple $\bu$ is inserted into relation $R_h$ at timestamp $\eta$ with weight $p$. We follow the set semantics, so inserting a tuple into a relation that already has it has no effect. Consider the stream of input tuples, ordered by their timestamp. Let $Q^\eta$ be the join defined by the first $\eta$ tuples of the stream, and set $Q^0 = \emptyset$. Also, let $\{\fp_i^\eta\}_{i\in\{1,2,\ldots,k\}}$ be the associated weight functions for $Q^\eta$. Each join result $\bu \in Q^\eta$ has its probability $\fp^\eta(\bu)$ defined similar to (\ref{eq:join-probability}). We use $\In$ to denote the length of the stream, which is only used in the analysis. The algorithms will not need the knowledge of $\In$, so they work over an unbounded stream. We have the corresponding versions of both problems in the dynamic setting: 
\boxminipg{0.9\linewidth}{
\begin{problem}[Dynamic Indexed Subset Sampling over Joins]
\label{pro:dynamic-index}
    Suppose tuples come in a streaming fashion. The goal is to maintain an index for answering subset sampling queries to the instance $\left<\join(Q^\eta),\fp^\eta\right>$ for every timestamp $\eta\in\mathbb{Z}^+$. Additionally, the random subsets returned by distinct queries must be independent.
\end{problem}
}

\boxminipg{0.9\linewidth}{
\begin{problem}[Dynamic One-shot Subset Sampling over Joins]
\label{pro:dynamic-one-shot}
    Suppose tuples come in a streaming fashion. The goal is to maintain one subset sample from the instance $\left<\join(Q^\eta),\fp^\eta\right>$ for every timestamp $\eta\in\mathbb{Z}^+$.
\end{problem}
}
\noindent For the dynamic index, we are interested in the maintenance time $t_u$ of the index when a tuple is inserted, and the time $t_s$ for answering a relational subset sampling query from the index. 
For the dynamic one-shot problem, we care about the total running time. Again, any solution to Problem~\ref{pro:dynamic-index} yields a solution to Problem~\ref{pro:dynamic-one-shot} with total time $O(t_u \cdot \In + t_s)$.

\smallskip 
For all versions of the relational subset sampling problem, we study the data complexity~\cite{abiteboul1995foundations} and analyze the complexity in terms of the data-dependent quantities (such as input size $\In$ and expected size of subset samples $\mu_{\Psi}$), while taking the schema size of $Q$ (i.e., $|V|$ and $|E|$) as a constant.

\smallskip \noindent {\bf Computation Model.}
We discuss our algorithms in the \emph{real RAM} model of computation~\cite{BM75, PS85}. In particular, we assume that the following operations take constant time: (i)~accessing a memory location; (ii)~generating a random value from the standard uniform distribution $\unif{0}{1}$; and (iii)~performing basic arithmetical operations involving real numbers, such as addition, multiplication, division, comparison, truncation, and evaluating fundamental functions like log.
We denote by $\geo{p}$ the geometric distribution over $\{0, 1, \ldots\}$ with probability mass function $\Pr[X=k] = (1-p)^k p$. We denote by $\tgeo{p}{n}$ the geometric distribution conditioned on the value being strictly less than $n$, i.e., with support $\{0, 1, \ldots, n-1\}$.
Under the RAM model, we can generate a random value $x$ from $\geo{p}$ in $O(1)$ time~\cite{bringmann2012efficient} by setting $x=\lfloor\frac{\log \unif{0}{1}}{\log(1-p)}\rfloor$. Additionally, we can generate a random value $x$ from $\tgeo{p}{n}$ in $O(1)$ time~\cite{bringmann2012efficient} by setting $x=\lfloor\frac{\log(1-q\cdot \unif{0}{1})}{\log(1-p)}\rfloor$, where $q=1-(1-p)^n$.

\smallskip \noindent {\bf Math Conventions.} For an integer $x \in \mathbb{Z}^+$, the notation $[x]$ denotes the set $\{1, 2, \ldots, x\}$, and $\llbracket x\rrbracket$ denotes the set $\{0,1,\ldots,x\}$. Given an ordered set of elements $X$ and a pair of elements $x,x'\in X$, 
we use $x \prec x'$ 
to indicate that $x$ is smaller 
than $x'$, and $x \preceq x'$ 
to indicate that $x$ is no larger 
than $x'$. 
We use $\log x$ to denote $\log_2 x$.

\smallskip
To help understanding, we list the main notations used by this paper in Table~\ref{tab:notations-full}.

\subsection{Related Work}
\label{sec:related}

\noindent {\bf Uniform Sampling over Joins.} \citet{chaudhuri1999random} and \citet{olken1993random} first established the uniform sampling over join problem, which returns a uniform sample from the join results. Zhao et al.~\cite{zhao2018random} introduced a linear space index for drawing uniform independent samples from acyclic joins. Later, \citet{chen2020random} investigated the cyclic joins and their results were improved to be conditionally optimal by two independent works \cite{deng2023join, Kim2023guaranteeing}. Most recently, Dai et al.~\cite{DHY24} proposed a reservoir sampling framework for maintaining uniform samples over joins in streaming settings. 

\vgap

\noindent {\bf Non-uniform Sampling over Joins.} Non-uniform sampling has been extensively studied in the context of online aggregation (OLA). The ripple join~\cite{haas1999ripple} and its variants generalize nested-loop joins to incrementally estimate aggregates with confidence intervals. Similarly, wander join~\cite{li2016wander} uses random walks over the join graph to provide unbiased estimators for aggregates like SUM or COUNT. These approaches differ fundamentally from ours: they are designed to estimate scalar statistics using algorithmic probabilities that minimize variance, whereas our goal is to materialize a concrete subset of independent samples according to user-defined importance weights.

Despite the substantial body of work on sampling over joins, we are not aware of any existing approach that can be adapted to support subset sampling over joins, indicating the need for novel techniques. 
A straightforward approach to subset sampling over joins is to first materialize the join result and then build a standard index over all output tuples. However, the size of the join result can be orders of magnitude larger than that of the input database, making this approach prohibitively inefficient. More broadly, there is a growing line of work on solving optimization problems on join results without fully materializing the join~\cite{agarwal2024computing, esmailpour2024improved, chen2022coresets, curtin2020rk, moseley2021relational, merkl2025diversity, khamis2018ac, kumar2015learning, schleich2016learning,
abo2021relational, yang2020towards, cheng2019nonlinear, cheng2021efficient, schleich2019learning, carmeli2023tractable,  arenas2024towards, kara2024f}.

\subsection{Our Results}
\label{sec:our-results} 

Our main results are summarized in Table~\ref{tab:complexity-comparison}. In this paper, we focus on acyclic joins, and all of our algorithms can be extended to cyclic joins using the standard tree decomposition approach \cite{gottlob2014treewidth}, at the cost of increasing the input size.\footnote{The input size $N$ will be replaced by $N^{\textsf{fhtw}}$, where \textsf{fhtw} is the fractional hypertree width \cite{gottlob2014treewidth} of the schema graph $G_Q$.} 
Our contributions address three primary scenarios: 

\begin{itemize}[leftmargin=*]
    \item {\bf (Section~\ref{sec:static-index})} For the static setting, we design an index that can be constructed in $O(N \log N \log \log N)$ time and $O(N \log N)$ space, supporting subset sampling queries in expected $O(1 + \mu_{\Psi} \log N)$ time, where $N$ is the input size and $\mu_{\Psi}$ is the expected sample size.
    \item {\bf (Section \ref{sec:oneshot})} 
    Then, we present a one-shot algorithm that bypasses index construction to generate a single subset sample in $O(N \log^2 N + \mu_{\Psi})$ expected time.
    \item {\bf (Section \ref{sec:dynamic})} Finally, we extend our framework to the dynamic setting with insertions, 
    by applying and adapting the dynamic direct access index for the acyclic join in~\cite{DHY24}.
\end{itemize}
\begin{table}[t]
\centering
\begin{adjustbox}{max width=\textwidth}
\begin{tabular}{lcccc}
\toprule
\textbf{Static} & \textbf{Method} & \textbf{Preprocessing Time} & \textbf{Sampling Time} & \textbf{Space Usage} \\
\midrule
\multirow{2}{*}{Index} & Baseline & $O( N + |\join(Q)|)$ & $O(1 + \mu_\Psi)$ & $O(|\join(Q)|)$ \\
& {Our Algorithm} & $O(\In \log \In \log\log \In)$ & $O(1 + \mu_\Psi \log \In)$ & $O(\In \log \In)$ \\
\midrule
\multirow{2}{*}{One-shot} & Baseline & --- & $O( N + |\join(Q)|)$ & $O(|\join(Q)|)$ \\
& {Our Algorithm} & --- & $O(\In\log^2\In + \mu_\Psi)$ & $O(\In\log^2 \In +\mu_\Psi)$ \\
\midrule
\textbf{Dynamic} & \textbf{Method} & \textbf{Update Time} & \textbf{Sampling Time} & \textbf{Space Usage} \\
\midrule
Index & {Our Algorithm} & $O(\log^ 3N \log\log N)$ & $O(\mu_\Psi \log N)$ & $O(\In \log \In)$ \\
\midrule
One-shot & {Our Algorithm} & --- & $O(N\log^ 3N \log\log N + \mu_\Psi \log N)$ & $O(\In \log \In)$  \\
\bottomrule
\end{tabular}
\end{adjustbox}
\caption{Complexity of baseline and our algorithms on acyclic joins for product function. $\In$ is the input size, $|\join(Q)|$ is the join size, and $\mu_\Psi$ is the expected output size.}
\label{tab:complexity-comparison}
\vspace{-2em}
\end{table}

\smallskip
\noindent {\bf Remark 1.} In the main text, we assume the function $\mathcal{F}$ to be the product of the input weights. In Appendix~\ref{appendix:other-functions}, we show how other functions can be supported by simply adapting our algorithms. 

\subsection{Prior Results as Preliminaries}
\label{sec:previous-results}

\noindent {\bf Subset Sampling.}
The classic subset sampling problem with a given set of elements has been well studied~\cite{bringmann2012efficient,yi2023optimal, bhattacharya2024nearoptimal,huang2023subset}. A naive approach is to iterate through all $n$ elements and flip a biased coin for each, taking $O(n)$ time. So, the goal is to design algorithms with query time proportional to the expected output size $\mu_\Psi$, which can be significantly smaller than the total number of elements. 
In the dynamic setting, where elements and probabilities can change, any algorithm requires $\Omega(1+\mu_\Psi)$ query time and $\Omega(1)$ update time~\cite{yi2023optimal}.
In the static setting, \cite{bringmann2012efficient} provided the first index with an optimal query time of $O(1+\mu_\Psi)$. However, their index requires $O(\log^2 n)$ update time, which is suboptimal in the dynamic setting. Recently, \cite{bhattacharya2024nearoptimal,yi2023optimal,huang2023subset} optimally solved this dynamic problem with indexes that take expected query time $O(1+\mu_\Psi)$, $O(n)$ space and $O(1)$ update time.

\smallskip
\noindent {\bf Worst-Case Optimal Joins.} 
The AGM bound states that for any join instance $Q$ of size $\In$, the maximum number of join results is $\Theta(\In^{\rho^*})$~\cite{atserias2013size}, where $\rho^*$ is the fractional edge covering number of the schema graph $G_Q$. The worst-case optimal join algorithms can compute any join instance $Q$ of size $\In$ within $O(\In^{\rho^*})$ time \cite{ngo2018worst}. 

\smallskip
\noindent {\bf Direct Access for Joins.} The {\em direct access} problem for joins was first studied by \citet{carmeli2020answering}, which assumes a fixed ordering on the join results in $\join(Q)$ and asks for an index to return the join result in some specific position. For any acyclic join $Q$, an index can be built in $O(\In)$ time such that each direct access query can be answered in $O(\log \In)$ time \cite{zhao2018random,carmeli2020answering}. 

\section{Classic Subset Sampling Revisited}
\label{sec:revisited}
In this section, we revisit the classic subset sampling problem, which serves as the foundation for our relational algorithms. We rely on the concept of a \DA oracle for the input set $S$ (under some fixed ordering\footnote{The ordering can be arbitrary but must remain consistent across multiple invocations of this oracle.}). Given an integer $i\in[|S|]$, this oracle returns the $i$-th element of $S$ in $O(1)$ time. All missing proofs are given in Appendix~\ref{appendix:revisited}.

\subsection{Perfect-Uniform Subset Sampling}
A subset sampling problem instance $\Psi = \left< S,\fp\right>$ is {\em perfect-uniform} if $\fp(e)=p$ for any $e \in  S$, which is also simplified as $\Psi = \left< S,p\right>$. 
We have the following folklore result for perfect-uniform scenario:
\begin{lemma} [Perfect-Uniform Subset Sampling]\label{lm:uss}
    For a perfect-uniform subset sampling instance $\Psi = \left< S,p\right>$
    , there is an index $\D$ that can be built within  $O(|S|)$ space and time, using which each subset sampling query can be answered in $O(1+\mu_\Psi)$ expected time.
\end{lemma}



\begin{wrapfigure}{r}{0.44\textwidth}
    \vspace{-1em}
   \begin{minipage}{0.44\textwidth}
      \begin{algorithm}[H]
    \caption{\textbf{uss-vanilla}$(S, p)$}
    \label{alg:uss-vanilla}
    $\mathbf{X} \gets \emptyset$; $i \gets 0$\;
    \While{$i < |S|$}{
        $i \gets i + 1 + \geo{p}$\;
        Add $\DA(S,i)$ to $\mathbf{X}$\;
    }
    \Return $\mathbf{X}$\;
\end{algorithm}
\end{minipage}
 \ \ \ \ 
\begin{minipage}{0.44\textwidth}
\begin{algorithm}[H]
    \caption{\textbf{uss-advanced}$(S, p)$}
    \label{alg:uss-advanced}
    $\mathbf{X} \gets \emptyset$; $q \gets 1-(1-p)^{|S|}$\; 
    \If{$\unif{0}{1} \le q$}{
        $i \gets 1 + \tgeo{p}{|S|}$\;
        Add $\DA(S, i)$ to $\mathbf{X}$\;
        \While{$i < |S|$}{
            $i \gets i + 1 + \geo{p}$\;
            Add $\DA(S, i)$ to $\mathbf{X}$\;
        }
    }
    \Return $\mathbf{X}$\;
\end{algorithm}
\end{minipage}
  \vspace{-2em}
\end{wrapfigure}

The index $\D$ consists of an array storing the elements of $S$, which allows us to retrieve the $i$-th element in $O(1)$ time. Additionally, we precompute the probability $q = 1 - (1-p)^{|S|}$ that the resulting sample is non-empty. The construction takes $O(|S|)$ time and space. Given this index, there are two algorithms to achieve the query time complexity. Algorithm~\ref{alg:uss-vanilla} represents the standard approach using geometric jumps to skip over rejected elements. Algorithm~\ref{alg:uss-advanced}, however, utilizes the precomputed probability $q$ to introduce a preliminary step: it first tosses a coin to decide whether to sample at least one element from $S$. While both algorithms are optimal for a single instance, Algorithm~\ref{alg:uss-advanced} offers a special benefit when handling multiple perfect-uniform instances. Instead of invoking the sampling procedure for every instance, 
this structure allows us to first sample 
the subset of instances that yield non-empty results, and then perform subset sampling only within those. This strategy significantly improves the total time complexity in the batched setting, as will be detailed in Section~\ref{sec:batched-rejection}.


\subsection{Rejection-based Subset Sampling}

The perfect-uniform condition is too restrictive for general applications. We can relax this requirement using a rejection-based strategy without affecting the asymptotic running time. Given a general instance $\Psi = \langle S, \fp \rangle$, we define an upper bound $p_{\max} = \max_{e \in S} \fp(e)$. We construct the index described in Lemma~\ref{lm:uss} for the bounding perfect-uniform instance $\Psi^+ = \langle S, p_{\max} \rangle$. This preprocessing takes $O(|S|)$ time.

The query procedure first retrieves a subset sample $\mathbf{X}$ from $\Psi^+$ using the index (via Algorithm~\ref{alg:uss-vanilla} or Algorithm~\ref{alg:uss-advanced}). Then, for each element $e' \in \mathbf{X}$, we retain it in the final sample with probability $\fp(e')/p_{\max}$. The overall runtime is proportional to the size of the intermediate sample $\mathbf{X}$. There are two conditions under which this runtime is efficiently bounded:

\begin{lemma} [$\beta$-uniform Subset Sampling]\label{lm:nss}
    Given a subset sampling instance $\Psi = \left< S,\fp\right>$, if $\max_{e\in S}\fp(e)\le\beta\cdot\min_{e\in S}\fp(e)$ for some parameter $\beta \ge 1$, there is an index $\D$ that can be built within $O(|S|)$ space and time, using which each subset sampling query can be answered in $O(1+\beta\cdot\mu_\Psi)$ expected time.
\end{lemma}
\vspace{-0.5em}
\begin{lemma} [Light Subset Sampling]\label{lm:lpss}
    Given a subset sampling instance $\Psi = \left< S,\fp\right>$, if $\max_{e\in S}\fp(e) \le 1/| S|$, there is an index $\D$ that can be built within $O(|S|)$ space and time, using which each subset sampling query can be answered in $O(1)$ expected time.
\end{lemma}

\begin{algorithm}[t]
    \caption{\textbf{ss-rejected-batch}$(\{\Psi_i\}_{i\in[m]})$}
    \label{alg:ss-rejected-batch}
    \KwIn{Implicit access to the meta-index for $\fq$ and sub-indexes for $\{\Psi_i\}$.}
    \KwOut{A subset sample $\rdX$.}
    $\rdI \gets \textbf{ss-query}(\Psi_{\text{meta}})$ \tcp*{Query the meta-index}
    \ForEach{$i \in \rdI$}{
        $\rdX_i \gets \emptyset$\;
        $j \gets 1+\tgeo{p_i^+}{|S_i|}$ \tcp*{Simulate Algorithm~\ref{alg:uss-advanced} using sub-index}
        Add $\DA(S_i, j)$ to $\rdX_i$\;
        \While{$j < |S_i|$}{
            $j \gets j+1+\geo{p_i^+}$\;
            \lIf{$j \le |S_i|$}{Add $\DA(S_i, j)$ to $\rdX_i$}
        }
        \lForEach{$e \in \rdX_i$}{
            Remove $e$ from $\rdX_i$ with probability $1-\frac{\fp_i(e)}{p_i^+}$}
    }
    \Return $\bigcup_{i \in \rdI} \rdX_i$\;
\end{algorithm}
\vspace{-0.5em}
\subsection{Batched Rejection-based Subset Sampling}
\label{sec:batched-rejection}

Consider a general subset sampling problem instance $\Psi=\langle S,\fp\rangle$ that is partitioned into a set of $m$ disjoint sub-instances $\set{\Psi_i=\langle S_i,\fp\rangle}_{i\in[m]}$, where $S=\bigsqcup_{i\in[m]}S_i$. Each sub-instance has an upper bound $p_i^+$ on the maximum probability (i.e., $\max_{e\in S_i}\fp(e)\le p_i^+$) and is either light or $\beta$-uniform.

A naive approach would be to construct the index from Section 2.2 for each sub-instance and invoke the query procedure for each one. However, this would take $\Omega(m)$ time per query, as we must spend at least $O(1)$ time to check every sub-instance. To improve efficiency, we construct a composite index that supports a two-stage sampling process:
\begin{itemize}[leftmargin=*]
    \item {\bf Preprocessing phase:} First, for each sub-instance $\Psi_i$, we build the standard index for $\langle S_i, p_i^+ \rangle$ as described in Lemma~\ref{lm:uss}. Second, we construct a meta-index to identify which sub-instances are likely to contribute to the final sample. We define a probability function $\fq:[m]\to[0,1]$ where $\fq(i)=1-(1-p_i^+)^{|S_i|}$. Note that $\fq(i)$ is the probability that sampling from $S_i$ with uniform probability $p_i^+$ yields a non-empty set. We build an optimal subset sampling index (e.g., \cite{bhattacharya2024nearoptimal}) for the meta-instance $\Psi_{\text{meta}} = \langle [m], \fq \rangle$. The total preprocessing time is linear in $|S|$.
\item {\bf Query phase:}
Algorithm~\ref{alg:ss-rejected-batch}
proceeds in two stages. First, it queries the meta-index to obtain a random set of indices $\rdI \subseteq [m]$. Then, only for the selected indices $i \in \rdI$, it invokes the query procedure on the sub-instance $\Psi_i$ (using Algorithm~\ref{alg:uss-advanced} logic) and performs rejection sampling.
\end{itemize}

\begin{lemma} [Batched Subset Sampling Index]
\label{lm:batched-index}
    Given a partitioned instance as defined above, there is an index that can be built in $O(|S|)$ time and space. Using this index, Algorithm~\ref{alg:ss-rejected-batch} returns a valid subset sample in $O\left(1 + \sum_{i=1}^m |S_i| \cdot p_i^+\right)$ expected time.
\end{lemma}

\vspace{-1em}
\section{Indexed Subset Sampling over Joins}
\label{sec:static-index}

We start with Problem~\ref{prob:static-index} for the join instance $Q$ and the associated weight functions $\{\fp_i\}_{i\in[k]}$. We assume tuples in each relation $R_i \in Q$ are in some arbitrary fixed ordering, and relations in $Q$ are also in some arbitrary fixed ordering. All missing proofs of this section are given in Appendix~\ref{appendix:static-index}.

\subsection{First Attempt}
\label{sec:first-algorithm}
So far, we first assume that a \DA oracle is available for the input join instance $Q$ (under some fixed ordering), such that it receives an integer $i \in [|\join(Q)|]$, and returns the $i$-th element of $\join(Q)$. Similar to before, the ordering can be arbitrary but must remain consistent across multiple invocations of this oracle. 
Let $L=\lceil2\rho^*\log\In\rceil$, where $\rho^*$ is the fractional edge covering number of the schema graph $G_Q$. Note that $2^L \ge |\join(Q)|$. 

\smallskip\noindent{\bf Step 1: Partition input relations.}
For each relation $R_i\in Q$, we partition tuples into $L+1$ sub-relations  
$R_i^{\langle 0 \rangle}, R_i^{\langle 1 \rangle}, \ldots, R_i^{\langle L \rangle}$ based on their probabilities, where $R_i^{\langle j \rangle}=\{\bu\in R_i \mid 2^{-j-1}<\fp_i(\bu)\le2^{-j}\}$ for $j\in 
\llbracket L-1\rrbracket
$ and $R_i^{\langle L \rangle}=\{\bu\in R_i|\fp_i(\bu)\le2^{-L}\}$. Then, each combination $\bj = (j_1,j_2,\ldots,j_k) \in \llbracket L \rrbracket^k$ defines a sub-join of $Q$, denoted as:
$Q_{\bj} = \{R_1^{\langle j_1\rangle}, R_2^{\langle j_2\rangle},\ldots, R_k^{\langle j_k \rangle}\}.$
The join results of $Q_{\bj}$ are a disjoint partition of the join results of $Q$, i.e., $\join(Q) = \bigsqcup_{\bj \in \llbracket L \rrbracket^k} \join(Q_{\bj})$. To answer a query to the subset sampling instance $\Psi=\left<\join(Q),\fp\right>$, we can return the disjoint union of the subsets sampled from all sub-instances
$\Psi_{\bj}=\left<\join(Q_{\bj}),\fp\right>$ for $\bj\in \llbracket L \rrbracket^k$. There are $(L+1)^k$ sub-instances in total.

\smallskip\noindent{\bf Step 2: Apply subset sampling to all sub-instances.} 
We first point out the following observation on each sub-instance in such a partition: 

\begin{lemma}[Either Light or Near-uniform Sub-instance] 
\label{lem:uniform-light}
Let $\bj = (j_1,j_2,\ldots,j_k)$. If $\sum_{i\in[k]} j_i \ge L$, the instance $\Psi_{\bj}$ is light, and otherwise, it is $2^k$-uniform. 
\end{lemma}

To efficiently sample from these sub-instances, we adopt the composite index strategy from Section~\ref{sec:batched-rejection}. We treat the sub-joins $\{Q_{\bj}\}_{\bj \in \llbracket L \rrbracket^k}$ as the disjoint sub-instances. We define the meta-probability $\fq(\bj) = 1-(1-2^{-\sum_{i\in[k]} j_i})^{|\join(Q_{\bj})|}$ and construct the meta-index for the instance $\langle \llbracket L \rrbracket^k, \fq \rangle$. This allows us to skip empty or non-selected sub-instances efficiently.
To put it formally:
\begin{itemize}[leftmargin=*]
    \item {\bf Preprocessing phase:} We partition input tuples by weights. For each sub-instance $Q_{\bj}$, we build a \DA oracle (the sub-index) and compute the join size (e.g., \cite{zhao2018random}). Additionally, we build the meta-index for $\fq$ as described in Section~\ref{sec:batched-rejection}.
    \item {\bf Query phase:} We invoke Algorithm~\ref{alg:ss-rejected-batch} using the constructed meta-index and sub-indexes. The algorithm first samples a set of indices $\mathbf{I}$ using the meta-index, and then retrieves tuples from the selected sub-instances $Q_{\bj}$ via their \DA oracles.
\end{itemize}

\smallskip
Notice that we partition each relation into $L+1$ sub-relations, which in turn defines $(L+1)^k$ sub-instances. We build \DA index~\cite{zhao2018random} for each sub-instance separately, so the cost is bounded by the total size of these sub-instances, which is $O(\In (L+1)^{k-1}) = O(\In \log^{k-1}\In)$. In answering a query, a \DA oracle can be supported in $O(\log \In)$ time, and each sub-instance is $2^k$-uniform, and all sub-instances form a partition of the original join result, so the cost of subset sampling over all sub-instances is $O(1+\mu_\Psi\log \In)$  following Lemma~\ref{lm:nss}, by setting $\beta = 2^k$. 
Putting everything together, we obtain our first result:
\begin{theorem}
\label{thm:logk}
Given an acyclic join instance $Q$ and a set of associated weight functions $\{\fp_j\}_{R_j \in Q}$, there is an index $\D$ that can be built within  $O(\In \log^{k-1}\In)$ space and time, using which each subset sampling query can be answered in $O(1+\mu_\Psi\log \In)$ expected time.
\end{theorem}

Next, we improve upon this initial result by showing an optimized index with much smaller space and preprocessing time, while maintaining the same query time. 

\subsection{Optimized Index for Subset Sampling over Joins}
\label{sec:index-optimized}
The framework in Section~\ref{sec:first-algorithm} establishes feasibility but is not space-efficient because it treats every combination of weight buckets $\bj \in \llbracket L \rrbracket^k$ as a distinct sub-instance. We now present an optimized index that reduces space usage to $O(\In \log \In)$.

\smallskip\noindent{\bf Merging Sub-instances by Score.}
Recall from Lemma~\ref{lem:uniform-light} that the classification of a sub-instance $Q_{\bj}$ (as either light or near-uniform) depends solely on the sum of indices $\sum_{i=1}^k j_i$. This observation allows us to merge sub-instances. Instead of maintaining $(L+1)^k$ disjoint partitions, we group join results based on their total weight ``score.''

For any tuple $\bu$ in a relation $R_i$, let its \emph{score} be $\phi(\bu) = \lfloor - \log \fp_i(\bu) \rfloor$. For a join result $\bu \in \join(Q)$, its \emph{score} is the sum of component scores: $\bar{\phi}(\bu) = \sum_{i=1}^k \phi(\bu[\schema(R_i)])$.
We partition the join results $\join(Q)$ into disjoint buckets based on this score. Let $\mathcal{B}_\ell = \{\bu \in \join(Q) \mid \bar{\phi}(\bu)=\ell\}$ denote the set of join results with score $\ell$.
We treat these buckets as the disjoint sub-instances required by the batched framework in Section~\ref{sec:batched-rejection}. Specifically:
\begin{enumerate}[leftmargin=*]
    \item \textbf{Buckets $\ell < L$:} Each bucket $\mathcal{B}_\ell$ contains join results with probabilities in the range $(2^{-\ell-1}, 2^{-\ell}]$. Thus, $\mathcal{B}_\ell$ forms a $2$-uniform sub-instance with upper bound $p^+_\ell = 2^{-\ell}$.
    \item \textbf{Tail Bucket $\ell \ge L$:} We merge all join results with score $\ell \ge L$ into a single tail bucket $\mathcal{B}_{\ge L} = \bigsqcup_{\ell \ge L} \mathcal{B}_\ell$. Since any $\bu \in \mathcal{B}_{\ge L}$ has $\fp(\bu) \le 2^{-L} \le 1/|\join(Q)|$, this combined bucket is a \emph{light} sub-instance (Lemma~\ref{lm:lpss}) with upper bound $p^+_{\ge L} = 2^{-L}$.
\end{enumerate}

\noindent{\bf Algorithm Overview.}
We apply Algorithm~\ref{alg:ss-rejected-batch} to these $L+1$ sub-instances.
\vspace{-0.2em}
\begin{itemize}[leftmargin=*]
    \item {\bf Preprocessing phase:} We compute the size $|\mathcal{B}_\ell|$ for each $\ell < L$ to build the meta-index. For the tail bucket $\mathcal{B}_{\ge L}$, we do not maintain an exact count or index; we simply assign it a weight proxy of $\In^{\rho^*}$ (an upper bound on join size) for the meta-index, or handle it via a fallback mechanism since its selection probability is negligible.
    
    \item {\bf Query phase:} We invoke Algorithm~\ref{alg:ss-rejected-batch} with inputs $\{\langle \mathcal{B}_\ell, \fp \rangle \mid \ell < L\}$ and the tail bucket. 
    \begin{itemize}[leftmargin=*]
        \item If the meta-sampling selects a bucket $\mathcal{B}_\ell$ with $\ell < L$, we use a specialized \DA oracle (described below) to retrieve a given index in $\mathcal{B}_\ell$ in $O(\log \In)$ time.
        \item If the tail bucket $\mathcal{B}_{\ge L}$ is selected (which happens with very low probability), we materialize the necessary results from $\mathcal{B}_{\ge L}$ on the fly to support the retrieval.
    \end{itemize}
\end{itemize} 
The 
challenge is to efficiently support \DA for the buckets $\mathcal{B}_\ell$ ($\ell
< L
$) without materializing them. We next describe how this can be achieved.

\smallskip
\noindent {\bf Join Tree with Notations.} As mentioned, any acyclic join $Q$ has a join tree $\T$ such that (1) there is a one-to-one correspondence between relations in $Q$ and nodes in $\T$, and (2) for each attribute $x$, the set of nodes containing $x$ forms a connected subtree. With a slight abuse of notation, we also use $R_i$ to denote the corresponding node in $\T$ for relation $R_i \in Q$. For simplicity, we assume the nodes in $\T$ are ordered by the in-order traversal.
For node $R_i$, let $\mathcal{C}_i$ be the child nodes of $R_i$ if $R_i$ is not a leaf. For a node $R_j \in \C_i$, let $R_{\textsf{next}(j)}$ be the immediately next child node after $R_j$ in $\C_i$. For the largest child node $R_{j^*}$, we set $\textsf{next}(j^*)=\mathsf{null}$. For node $R_i$, we use $\textrm{parent}(i)$ to denote the (unique) parent node of $R_i$ if $R_i$ is not the root. Let $\key(i) = \schema(R_i) \cap \schema(R_{\textrm{parent}(i)})$ be the (common) join attributes between relation $R_i$ and its parent if $R_i$ is not the root. For the root node $r$, $\key(r) = \emptyset$. 
For a node $R_i$, let $\T_i$ be the subtree rooted at node $R_i$. For an arbitrary child $R_j \in \C_i$, we define $\T_i^j$ to be the subtree of $\T_i$ that excludes any subtrees rooted at the child nodes 
in
$\C_i$ coming before $R_{j}$, i.e., $\displaystyle{\T^j_i = \T_i \setminus \cup_{R_{j'} \in \C_i: R_{j'} \prec R_j} \T_{j'}}$.
For completeness, we also define $\T^\emptyset_i = \T_i$.

\smallskip
\noindent {\bf Data Structure.}
We first remove all dangling tuples from the database instance.
For each relation $R_i \in Q$, we store input tuples in a hash table, so that for any subset of attributes $U \subseteq \schema(R_i)$ and a tuple $t \in \dom(U)$, we can get the list of tuples $R_i \ltimes t$ in $O(1)$ time. We take an arbitrary join tree $\T$ for $Q$ and store the following statistics to augment $\T$.

\smallskip
\noindent {\bf $\mathbf{W}$-values.}
Consider any internal node $R_i \in Q$.
For each tuple $\bu \in R_i$, we store $L$ counters $W^j_{i,\bu}(\ell)$ with $\ell \in \llbracket L-1 \rrbracket$ specifically for each child node $R_j \in \C_i$, where $W^j_{i,\bu}(\ell)$ stores the number of join results produced by relations in the subtree $\T^j_i$ that is also participated by $\bu$, with score $\ell$, i.e., 
    \begin{equation}
    \vspace{-0.5em}
        \label{eq:W}
        W^j_{i,\bu}(\ell) = \left|\left\{\bm{t} \in (\Join_{R_h \in \T^j_i} R_h)\ltimes \bu: \ \sum_{R_h \in \T^j_i} \phi(\bm{t}[\schema(R_h)]) = \ell\right\}\right|.
    \end{equation}
    Furthermore, we define and store $W^\emptyset_{i,\bu}(\ell)$, the number of join results produced by relations in the subtree $\T_i$ that is also participated by $\bu$, with score $\ell$, by replacing $\T^j_i$ in (\ref{eq:W}) with $\T^\emptyset_i$. For completeness, for the largest node $R_{j^*} \in \C_i$, we also define $W^{\textrm{next}(j^*)}_{i,\bu}(0) = 1$ and $W^{\textrm{next}(j^*)}_{i,\bu}(\ell) = 0$ for every $\ell \in [L]$.
    Moreover, for each score $\ell \in \llbracket L-1\rrbracket$, and each child node $R_j \in \C_i$, we build a prefix-sum array on $\langle W^j_{i, \bu}(\ell): \bu \in R_i\rangle$ under the pre-determined ordering of tuples in $R_i$. For each leaf node $R_i \in Q$, we only store the counter $W^\emptyset_{i,\bu}(\ell)$ for each tuple $\bu \in R_i$.

\smallskip
\noindent {\bf $\mathbf{M}$-values.} For a non-root node $R_i$, we compute the projection of $R_i$ onto $\key(i)$, as $R_i[\key(i)]$. For each tuple $\bm{v} \in R_i[\key(i)]$, we also compute $M_{i,\bv}(0), M_{i,\bv}(1), \ldots, M_{i,\bv}(L-1)$: for $\ell \in \llbracket L-1 \rrbracket$, $M_{i,\bv}(\ell)$ stores the sum of counters $W^\emptyset_{i,\bu}(\ell)$ over all tuples $\bm{u} \in R_i \ltimes \bv$, i.e.,
    \begin{equation}
    \label{eq:M}
    M_{i,\bv}(\ell) = \sum_{\bu \in R_i \ltimes \bv} W^\emptyset_{i,\bu}(\ell).
    \end{equation}

   We briefly describe how to preprocess the data structures (see Appendix~\ref{appendix:static-index} for complete details).
   We compute $W$/$M$ values in a bottom-up way. $M$-values will be computed according to (\ref{eq:M}) straightforwardly once $W$-values are well computed. Consider a leaf node $R_i$, we have $W^\emptyset_{i,\bu}(\phi(\bu))=1$ and $W^\emptyset_{i,\bu}(j)=0$ for any $j \in \llbracket L-1 \rrbracket\setminus\{\phi(\bu)\}$. Consider an internal node $R_i$. Suppose $W$-values and $M$-values are well defined for each child $R_j \in \mathcal{C}_i$. We compute $W^{j}_{i, \bu}$ in a decreasing ordering of nodes in $\C_i$ and follow a recursive structure. If $\ell < \phi(\bu)$, $W^{j}_{i, \bu}(\ell) = 0$, and otherwise, 
    \begin{align}
    \label{W-recursive}
        W^j_{i,\bu}(\ell) = \sum_{\substack{(\ell_1,\ell_2) \in \llbracket \ell - \phi(\bu) \rrbracket: \ell_1 + \ell_2 = \ell - \phi(\bu)}} M_{j,\bu[\key(j)]}(\ell_1)\cdot W^{\textsf{next}(j)}_{i, \bu}(\ell_2).
    \end{align}
    Note that $W^\emptyset_{i,\bu}(\ell)$ can be computed similarly by resorting to the first child node. Importantly, the computation of (\ref{W-recursive}) can be sped up by Fast Fourier Transform~\cite{cooley1965algorithm, cormen2022introduction}, and the whole computation only takes $O(\In L\log L)$ time. Lastly, the bucket size $|\B_\ell|$ for each score $\ell \in \llbracket L-1\rrbracket$ can be rewritten as $|\B_\ell| = \sum_{\bu \in R_r} W^\emptyset_{r,\bu}(\ell)$ for the root node $r$. A pseudo-code for this stage is shown in Algorithm~\ref{alg:preprocess}.

    \begin{algorithm}[t]
    \caption{\textsc{\bf RecursiveAccess}$(i, j, \bv, \ell, \tau)$}
    \label{alg:fast-recursive}
    \KwIn{$i,j\in [k]$ such that $R_j \in \C_i$, tuple $\bv$, score $\ell \in \llbracket L-1 \rrbracket$, integer $\tau \in \mathbb{Z}^+$.}
    \KwOut{The $\tau$-th tuple in the join result of $\T^j_i$, with score $\ell$, that can be joined with $\bv$.}
    $\bu \gets$ the smallest tuple in $R_i \ltimes \bv$ such that $\sum_{\bu' \in R_i \ltimes \bv: \bu'\preceq\bu} W^j_{i,\bu'}(\ell) \ge \tau$ \label{locate-u}\;
    \lIf{$R_i$ is a leaf node}{\Return $\bu$\label{base-case}}
    $\tau \gets \tau - \sum_{\bu' \in R_i \ltimes \bv: \bu' \prec \bu} W^j_{i,\bu'}(\ell)$ \label{correct-index-u}\; 
    $\Phi \gets \{(\ell_1, \ell_2) \in \llbracket \ell \rrbracket^2: \ell_1 + \ell_2 =\ell - \phi(\bu)\}$ in lexicographical ordering by $(\ell_1, \ell_2)$ \label{phi}\;
    $(\ell_1, \ell_2) \gets$ the smallest pair in $\Phi$ that $\displaystyle{\sum_{(\ell'_1,\ell'_2) \in \Phi: (\ell'_1,\ell'_2) \preceq (\ell_1, \ell_2)} M_{j, \bu[\key(j)]}(\ell_1') \cdot W^{\textsf{next}(j)}_{i, \bu}(\ell_2') \ge \tau}$ \label{locate-score-pairs}\;
    $\tau \gets \tau - \displaystyle{\sum_{(\ell'_1,\ell'_2) \in \Phi: (\ell'_1,\ell'_2) \prec (\ell_1, \ell_2)} M_{j, \bu[\key(j)]}(\ell'_1) \cdot W^{\textsf{next}(j)}_{i, \bu}(\ell'_2)}$ \label{correct-index-score-pairs}\;
    $\tau_1 \gets \left\lceil \dfrac{\tau}{W^{\textsf{next}(j)}_{i, \bu}(\ell_2)} \right\rceil$ \label{left-index}, $\tau_2 \gets ((\tau-1) \bmod W_{i,\bu}^{\textrm{next}(j)}(\ell_2)) + 1$ \label{right-index}\;
    $\bu_1 \gets \textsc{\bf RecursiveAccess}(j, \emptyset, \bu, \ell_1, \tau_1)$\label{call-one}\;
    $\bu_2 \gets \textsc{\bf RecursiveAccess}(i, \textrm{next}(j), \bu, \ell_2, \tau_2)$\label{call-two}\;
    \Return $\bu_1 \Join \bu \Join \bu_2$\;
\end{algorithm}

\noindent \textbf{Algorithm for \textsf{DirectAccess}.}
Our goal is to retrieve the $\tau$-th tuple from bucket $\mathcal{B}_\ell$ given score $\ell \in \llbracket L-1\rrbracket$ and rank $\tau \in [|\mathcal{B}_\ell|]$. To ensure the ``$\tau$-th'' tuple is well-defined, we rely on the canonical lexicographic ordering of join results. This ordering is naturally induced by the fixed traversal order of relations in the join tree $\T$ and the total ordering of tuples within each relation $R_i$.

\smallskip

We present a general procedure \textsc{\bf RecursiveAccess}$(i, j, \bv, \ell, \tau)$, that returns the $\tau$-th tuple in the join results of the relations in the subtree $\T_i^j$, with score $\ell$, that can be joined with $\bv$ on its support attributes. The original problem can be solved by invoking {\bf RecursiveAccess}$(r, \emptyset, \emptyset, \ell, \tau)$.

Algorithm~\ref{alg:fast-recursive} begins by identifying the specific tuple $\bu \in R_i$ that participates in the target join result. Line~\ref{locate-u} performs a binary search on the prefix-sum stored over $W^j_{i,*}(\ell)$ to locate the first tuple $\bu$ where the cumulative count of join results in the subtree $\T^i_j$ (with score $\ell$) covers the rank threshold $\tau$. Line~\ref{correct-index-u} then subtracts the contribution of tuples preceding $\bu$ from $\tau$. This reduces the problem to retrieving the (updated) $\tau$-th join result within the subtree $\T^i_j$ specifically involving $\bu$. We then distinguish two cases. If $R_i$ is a leaf node, we simply return $\bu$. In the general case, we decompose the retrieval into $\T_j$ and $\T^{\textrm{next}(j)}_i$ based on the fact that $\T^{j}_i = \T_j \bigsqcup \T^{\textrm{next}(j)}_i$. The objective is to determine the local indexes to query within these sub-problems. Recall that $W^{j}_{i,\bu}(\ell)$ aggregates disjoint contributions from score pairs $(\ell_1, \ell_2)$ (Equation~\ref{W-recursive}). The algorithm iterates through these pairs in lexicographical order (Line~\ref{phi}) to find the specific pair $(\ell_1, \ell_2)$ that covers the current rank threshold $\tau$ (Line~\ref{locate-score-pairs}). After further adjusting $\tau$ to account for preceding pairs (Line~\ref{correct-index-score-pairs}), we effectively map the remaining rank to indexes $\tau_1$ and $\tau_2$ within the Cartesian product of valid results from $\T_j$ and $\T^{\textrm{next}(j)}_i$ (Line~\ref{left-index}). Finally, we invoke the procedure recursively on $\T_j$ and $\T^{\textrm{next}(j)}_i$ and combine the returned results.

We provide the complete analysis in Appendix~\ref{appendix:static-index}. At last, we obtain:
\begin{theorem}
\label{thm:best}
Given an acyclic join instance $Q$ of size $\In$, and a set of associated weight functions $\{\fp_j\}_{R_j \in Q}$, there exists an index $\D$ of size $O(\In \log\In)$ that can be built in  $O(\In \log \In \log\log \In)$ time, such that each subset sampling query can be answered in $O(1+\mu_\Psi\log \In)$ expected time.
\end{theorem}

\section{One-shot Subset Sampling over Joins}
\label{sec:oneshot}
\renewcommand{\P}{\mathcal{P}}

We observe that the data structure from Theorem~\ref{thm:best} provides a baseline solution for the one-shot problem. By constructing the index and issuing a single query, we can generate a subset sample in $O(\In \log \In \log\log\In + \mu_\Psi \log \In)$ expected time. In this section, we present a specialized algorithm for the one-shot setting that runs in $O(\In \log^2\In + \mu_\Psi)$ expected time. This approach improves upon the baseline solution when the expected sample size is large (i.e., $\mu_\Psi \gg \In$), as it eliminates the logarithmic factor associated with $\mu_\Psi$. All omitted proofs and details are provided in Appendix~\ref{appendix:oneshot}.

\smallskip
\noindent{\bf High-level Idea.} We first keep all statistics computed in the preprocessing phase as in Section~\ref{sec:index-optimized}. We conceptually partition the join results into buckets. Recall that the high-level idea of answering one subset sampling query is to invoke Algorithm~\ref{alg:ss-rejected-batch} on instances $\{\langle \mathcal{B}_\ell, \fp \rangle: \ell \in \llbracket L-1\rrbracket\}$ with $p^+_{\ell} = 2^{-\ell}$, and $\langle \bigsqcup_{\ell\ge L} \mathcal{B}_\ell, p \rangle$ with $p^+_{L} = 2^{-L}$. Whenever a join result needs to be retrieved from  $\mathcal{B}_\ell$, if $\ell < L$, we use the \DA oracle to support the retrieval, and otherwise, we compute all the results in $\bigsqcup_{\ell\ge L} \mathcal{B}_\ell$ to support the retrieval. In the first case, whenever the \DA oracle on bucket $\mathcal{B}_\ell$ with $\ell < L$ is invoked, it pays $O(\log \In)$ time for each retrieval. To overcome this barrier in the query time, we need to compute additional statistics. 
Careful inspection of Algorithm~\ref{alg:fast-recursive} reveals that the $O(\log \In)$ cost comes from the binary search in Line \ref{locate-u} and the exhaustive search in Line \ref{locate-score-pairs}. A natural idea is to precompute all these statistics. Consider an arbitrary relation $R_i \in Q$ with an arbitrary child node $R_j \in \C_i$. For every score $\ell \in \llbracket L-1 \rrbracket$, and every tuple $\bv \in R_i[\key(i)]$, we compute an array $X_{i,j,\bv,\ell}$ of size $|R_i \ltimes \bv|$, where for every tuple $\bu \in R_i \ltimes \bv$, $$X_{i,j, \bv, \ell}(\bu) = \sum_{\bu' \in R_i \ltimes \bv,\, \bu' \preceq \bu} W^j_{i,\bu'}(\ell).$$ 
Next, for every score $\ell \in \llbracket L-1 \rrbracket$ and every tuple $\bu \in R_i$, we compute an array $Y_{i,j,\bu,\ell}$ of size at most $L$, where for every pair $(\ell_1,\ell_2)\in \llbracket \ell\rrbracket^2$ with $\ell_1 + \ell_2 = \ell - \phi(\bu)$, in lexicographically sorted order, $$Y_{i,j,\bu, \ell}(\ell_1,\ell_2) = \sum_{(\ell_1',\ell_2')\in \llbracket \ell \rrbracket: (\ell_1',\ell_2') \preceq (\ell_1,\ell_2), \ell_1' + \ell_2' = \ell - \phi(\bu)} M_{j,\bu[\key(j)]}(\ell_1') \cdot W_{i,\bu}^{\textrm{next}(j)}(\ell_2').$$ 
We note that by constructing $X$ and $Y$ arrays, we do not define a new index for subset sampling over joins. Instead, these statistics
are computed when answering one-shot subset sampling query.

Each invocation of Algorithm~\ref{alg:fast-recursive} is identified by a quintuple $(i,j,\bv,\ell,\tau)$ as its input. 
The high-level idea is to maintain the list of quintuples $(i,j,\bv,\ell,\tau)$ on which we will invoke Algorithm~\ref{alg:fast-recursive}, and run multiple invocations of Algorithm~\ref{alg:fast-recursive} together by merging certain operations. 
We group these quintuples by $(i,j,\bv,\ell)$, where quintuples inside each group only have different $\tau$-values sorted by $\tau$ in increasing order. We maintain the set of quintuples $\calP_{i,j}$ we should satisfy in each call of   Algorithm~\ref{alg:fast-recursive}. For a node $R_i$ and $R_j\in\C_i\cup\{\emptyset\}$, $\calP_{i,j}$ is defined as the parameter set of all invocations of Algorithm~\ref{alg:fast-recursive} on $(i,j)$.
Initially, for each request $\DA(S_\ell,\tau)$ issued in
Algorithm~\ref{alg:ss-rejected-batch},  
we add a quintuple $(r,\emptyset, \emptyset, \ell, \tau)$ to $\calP_{r,\emptyset}$, where $R_r$ is the root node. 
Below, we describe a general procedure {\bf BatchRecursiveAccess}$(i, j, \P_{i,j})$. 
The new algorithm executes the same steps as in Algorithm~\ref{alg:fast-recursive}, but it computes all smallest tuples $\bu$ and pairs $(\ell_1,\ell_2)$ for all quintuples simultaneously. 

\smallskip
\noindent {\bf {\bf BatchRecursiveAccess}$(i, j, \P_{i,j})$:} This procedure takes as input a node $R_i \in Q$ with a child node $R_j \in \C_i$, and the set of quintuples $\P_{i,j}$.
The output will be exactly the collection of tuples returned by Algorithm~\ref{alg:fast-recursive} for each invocation parameterized by quintuple $(i,j,\bv,\ell,\tau) \in \P_{i,j}$.
The pseudocode is shown in Algorithm~\ref{alg:one-shotv2}.

We first sort $\P_{i,j}$ based on $\tau$ using the Radix sort algorithm~\cite{cormen2022introduction}.
Since $\tau$ is polynomially bounded in $\In$, using radix sort saves a factor of $\log \In$ in the sorting time.
Then we group $\P_{i,j}$ by the different $(i,j,\bv, \ell)$ creating the groups $\P_{i,j,\bv,\ell}$. We notice that all quintuples in each $\P_{i,j,\bv,\ell}$ are sorted by $\tau$ using the sorted ordering of $\P_{i,j}$. 
The goal is to compute the smallest tuple $\bu$ for each quintuple, as in Line~\ref{locate-u} of Algorithm~\ref{alg:fast-recursive}. 
For every score $\ell\in \llbracket L-1\rrbracket$, we consider every quintuple $(i,j,\bv,\ell,\tau)\in \P_{i,j,\bv,\ell}$ following the sorted order. We traverse table $X_{i,j,\bv,\ell}$ until we get the first tuple $\bu$ such that $X_{i,j,\bv,\ell}(\bu)\geq \tau$. We also update the quintuple $(i,j,\bv,\ell,\tau)$ to $(i,j,\bv,\ell,\tau-X_{i,j,\bv,\ell}(\mathsf{prev}(\bu)))$ as in Line~\ref{correct-index-u} of Algorithm~\ref{alg:fast-recursive}. 
If $R_i$ is a leaf node we return all smallest tuples $\bu$ we computed from every quintuple in $\P_{i,j}$. For later reference, we denote the returned tuple $\bu$ for the quintuple $(i,j,\bv,\ell,\tau)$ as $\vec{\bu}(i,j,\bv,\ell,\tau)$.

Then we proceed in a similar manner to compute the pairs $(\ell_1,\ell_2)$.
We sort (the updated) $\P_{i,j}$ by $\tau$ using the Radix sort algorithm. Then we group $\P_{i,j}$ by the different $(i,j,\bv, \ell)$ creating the groups $\P_{i,j,\bv,\ell}$.
The goal is to compute the smallest pair $(\ell_1,\ell_2)$ for each quintuple, as in Line~\ref{locate-score-pairs} of Algorithm~\ref{alg:fast-recursive}. For every score $\ell\in \llbracket L-1\rrbracket$ we consider every quintuple $(i,j,\bv,\ell,\tau)\in \P_{i,j,\bv,\ell}$ following the sorted order. We traverse table $Y_{i,j,\bv,\ell}$ until we get the first pair $(\ell_1,\ell_2)$ such that $Y_{i,j,\bv,\ell}(\ell_1,\ell_2)\geq \tau$.
We compute $\tau_1, \tau_2$ as in Line~\ref{left-index} of Algorithm~\ref{alg:fast-recursive}. Furthermore, we add the quintuple $(j,\emptyset,\bu,\ell_1,\tau_1)$ to the (initially empty) set $\P_{j,\emptyset}$
and the quintuple $(i,\mathsf{next}(j),\bu,\ell_2,\tau_2)$ in (initially empty) set $\P_{i,\mathsf{next}(j)}$.
Notice that these sets are used as input to the next calls $\textsc{\bf BatchRecursiveAccess}(j, \emptyset, \P_{j,\emptyset})$ and $\textsc{\bf BatchRecursiveAccess}(i, \mathsf{next}(j), \P_{i,\mathsf{next}(j)})$ as in Lines~\ref{call-one} and~\ref{call-two} of Algorithm~\ref{alg:fast-recursive}.

Now, we will generate tuples to be returned for quintuples in $\P_{i,j}$. Suppose the sets of tuples returned by the two recursive invocations are available.
For every quintuple $(i,j,\bv,\ell,\tau)\in \P_{i,j}$, let $\vec{\bu}(j,\emptyset,\bu,\ell_1,\tau_1)$ be the tuple returned for the quintuple $(j,\emptyset,\bu,\ell_1,\tau_1)\in \P_{j,\emptyset}$, and $\vec{\bu}(i,\textrm{next}(j),\bu,\ell_2,\tau_2)$ be the tuple returned for the quintuple $(i,\textrm{next}(j),\bu,\ell_1,\tau_1)\in \P_{i,\textrm{next}(j)}$,
We simply join tuple $\vec{\bu}(j,\emptyset,\bu,\ell_1,\tau_1)$, and tuple $\vec{\bu}(i,\textrm{next}(j),\bu,\ell_2,\tau_2)$ together with tuple $\bu$ to form the tuple $\vec{\bu}(i,j,\bu,\ell,\tau)$. Finally, we will return all tuples for each quintuple $(i,j,\bv,\ell,\tau)\in \P_{i,j}$.  

\smallskip
\smallskip 

\vspace{-0.5em}
\begin{theorem}
\label{thm:bestOneShot}
For an acyclic join instance $Q$ of size $\In$, and a set of associated weight functions $\{\fp_j\}_{R_j \in Q}$, there exists an algorithm that returns one subset sample of the instance $\Psi=\langle\join(Q),\fp\rangle$ in $O\left(\min\left\{ \In \log \In \log \log \In + \mu_\Psi \log \In, \: \In \log^2 \In + \mu_\Psi \right\}\right)$ expected time.
\end{theorem}
\vspace{-0.5em}


\section{Dynamic Subset Sampling over Joins}
\label{sec:dynamic}
In this section, we move to the dynamic setting, where tuples are inserted one by one. We first adapt the techniques of~\cite{DHY24} to directly extend our non-optimal (static) index from Theorem~\ref{thm:logk} to a dynamic variant that maintains an approximate \textsf{DirectAccess} index. We then introduce new techniques and ideas to extend our optimized static index from Theorem~\ref{thm:best} to the dynamic setting. Throughout, we assume that, for each relation, tuples are ordered by their insertion timestamps.

\subsection{First Attempt} 
\label{sec:approximate-DA}
The high-level idea is to approximately maintain the static index presented in Section~\ref{sec:static-index}. However, as pointed by \cite{DHY24}, even maintaining the join size for line-3 join $Q = \{R_1(X,Y), R_2(Y,Z), R_3(Z,W)\}$ requires at least $\Omega(\sqrt{N})$ time for each insertion, which is too costly in the streaming setting. Hence, we follow the approach proposed by \cite{DHY24} by maintaining an approximate \DA index more efficiently while not sacrificing the overall sampling efficiency asymptotically. 
We first define the delta query $\Delta \join(Q, \bu)$ as the set of new join results generated when tuple $\bu$ is inserted into a relation of $Q$. Formally, if $\bu$ is inserted into $R_i$, $\Delta \join(Q, \bu) = \join(Q \setminus \{R_i\} \cup \{R_i \cup \{\bu\}\}) \setminus \join(Q)$.

\begin{theorem}[\cite{DHY24}]
\label{the:acyclic}
Given an initially empty acyclic join $Q$, we can maintain an index $\mathcal{L}$ on $Q$ that supports the following operations:
\begin{itemize}[leftmargin=*]
    \item When a tuple $t$ is added to $Q$, $\mathcal{L}$ can be updated in amortized $O(\log N)$ time.
    \item The index defines an array $J \supseteq \join(Q)$ where the tuples in $\join(Q)$ are the real tuples and the others are dummy. 
    The index can return $|J|$ in $O(1)$ time. For any given $j\in [|J|]$, it returns $J[j]$ in $O(\log N)$ time. Furthermore, $J$ is guaranteed to be $\lambda$-dense\footnote{A stream $S = \langle x_1, x_2, \ldots, x_n \rangle$ is \textit{$\lambda$-dense} for $0< \lambda \le 1$, if $r_i \ge \lambda \cdot (i-1)$ for all $i$, where $r_i$ is the number of real (non-dummy) items in the first $i-1$ items.} for some constant probability $0<\prob\le 1$.
    \item The above also holds when $\join(Q)$ is replaced by the delta query $\Delta \join(Q, \bu)$ for any $\bu \not\in Q$.
\end{itemize}
where $N$ is the total number of insertions, but the index does not need the knowledge of $N$ in advance. 
\end{theorem}

Using Theorem~\ref{the:acyclic}, we derive a simple dynamic algorithm by adapting the framework from Section~\ref{sec:first-algorithm}. 
As in Section~\ref{sec:first-algorithm}, we define  $(L+1)^k$ sub-instances $Q_{\bj}$ for $\bj \in \llbracket L \rrbracket^k$.
In the dynamic setting, we maintain a dynamic index $\mathcal{L}_{\bj}$ (as described in Theorem~\ref{the:acyclic}) for \emph{each} sub-instance $Q_{\bj}$. 
When a tuple $t$ is inserted into relation $R_i$, we first determine its bucket index $a$ such that $t \in R_i^{\langle a \rangle}$. This tuple $t$ participates in $(L+1)^{k-1}$ sub-instances. 
We update the corresponding $(L+1)^{k-1}$ indexes. Since each update takes $O(\log \In)$ time amortized, the total amortized update time is $O(L^{k-1} \log \In)$.
To answer a subset sampling query, we use the batched rejection sampling strategy. For each sub-instance $Q_{\bj}$, the index $\mathcal{L}_{\bj}$ provides the size of the implicit array $|J_{\bj}|$. We use $|J_{\bj}|$ as an upper bound on the size of the sub-join to define the meta-sampling probability $\fq(\bj)$. Suppose a sub-instance $Q_{\bj}$ is selected and the algorithms needs to retrieve the tuple $u = J_{\bj}[z]$. If $u$ is a dummy tuple (which happens with probability at most $1-\prob$) we reject it. If $u$ is real, we proceed with the standard rejection check based on its weight. Since $\prob$ is a constant, the expected query time remains asymptotically the same as the static case.
\vspace{-0.2em}
\begin{corollary}
\label{cor:basic-dynamic}
Given an acyclic join instance $Q$ and a set of associated weight functions $\{\fp_j\}_{R_j \in Q}$, we can maintain a dynamic index under tuple insertions using $O(\In \log^{k-1}\In)$ space, with amortized update time $O(\log^k \In)$, while supporting answering the subset sampling query in $O(1 +  \mu_\Psi \log \In)$ expected time, where $N$ is the total number of insertions. 
\end{corollary} 
\vspace{-0.5em}

\subsection{Optimized Dynamic Index for Subset Sampling over Joins}
\label{sec:dynamic-optimized-index}
To improve both space and update time, we adapt the optimized index from Section~\ref{sec:index-optimized} to the dynamic setting. This adaptation is nontrivial, since the index in Section~\ref{sec:index-optimized} is inherently static and therefore requires new ideas. Rather than maintaining disjoint indexes for every bucket combination, we maintain the aggregated statistics (the $W$ and $M$ values) directly. Because maintaining exact counts is too costly, we instead store \emph{approximate} statistics, rounded up to the next power of $2$.

\begin{algorithm}[t]
    \caption{\textbf{Update}$(i, \bv, \ell, \Delta)$}
    \label{alg:update}
    \KwIn{$i \in [k]$, tuple $\bv \in \dom(\key(i))$, score $\ell$, and 
    $\tilde{W}^\emptyset_{i, \bu}(\ell)$ increased by $\Delta$ for $\bu \in R_i \ltimes \bv$.}
    Update the prefix-sum tree built for $\{\tilde{W}_{i,\bu'}^\emptyset(\ell): \bu' \in R_i,\bu'[\key(i)] = \bu[\key(i)]\}$\label{line:prefsum}\;
    $\hat{M}_{i,\bv}(\ell) \gets \hat{M}_{i,\bv}(\ell) + \Delta$\label{line:dynM}\; 
    $\displaystyle{\tilde{M}_{i,\bv}(\ell) \gets 2^{\left\lceil\log \hat{M}_{i,\bv}(\ell)\right\rceil}}$\;
    \If{$\tilde{M}_{i,\bv}(\ell)$ changes and $\textsf{parent}(i) \neq \mathsf{null}$}{
        $p \gets \textsf{parent}(i)$\;
        \ForEach{tuple $\bu \in R_{p} \ltimes \bv$}{
             Compute $\Tilde{W}^j_{p, \bu}(\ell')$ for each score $\ell'\in \llbracket L -1\rrbracket$ and $R_j\in \C_p$ according to (\ref{eq:tilde-W})\label{line:dynW}\; 
             \ForEach{$\ell'\in \llbracket L -1\rrbracket$}{
             \If{$\Tilde{W}^\emptyset_{p, \bu}(\ell')$ is increased (by $\Delta'$)}{
                $\textbf{Update}(p, \bu[\key(p)], \ell',\Delta')$\;
             }
             }
        }
    }
\end{algorithm}
\smallskip
 
For each relation $R_i \in Q$, each child node $R_j \in \C_i$, and each tuple $\bu \in R_i$, we maintain an upper bound $\tilde{W}^j_{i,\bu}(\ell)$ on $W^j_{i,\bu}(\ell)$. Similarly, for each non-root node $R_i$, tuple $\bv \in R_i[\key(i)]$, and score $\ell \in \llbracket L-1 \rrbracket$, we 
maintain $\hat{M}_{i,\bv}(\ell)$ as the sum of the approximate counters:
\begin{equation}
    \label{eq:M-dynamic}
    \hat{M}_{i,\bv}(\ell) = \displaystyle{\sum_{\bu \in R_i\ltimes\bv} \tilde{W}^\emptyset_{i,\bu}(\ell)},
    \vspace{-0.3em}
\end{equation}
and use it to compute the upper bound $\displaystyle{\tilde{M}_{i,\bv}(\ell) = 2^{\left\lceil\log \hat{M}_{i,\bv}(\ell)\right\rceil}}$.
Consider a leaf node $R_i$. As $\C_i = \emptyset$, we have $\tilde{W}^{\emptyset}_{i, \bu}(\ell) = 1$ if $\ell = \phi(\bu)$, and $\tilde{W}^{\emptyset}_{i, \bu}(\ell) = 0$ otherwise. Consider an internal node $R_i$. Suppose $\tilde{W}$ and $\tilde{M}$-values have been computed for each child node $R_j \in \C_i$. We compute $\tilde{W}^{j}_{i, \bu}$ in a decreasing ordering of nodes in $\C_i$. Let $R_{\textsf{next}(j)} \in \C_i$ be the immediate next sibling of $R_j$.
\begin{equation}
    \label{eq:tilde-W}
    \tilde{W}^{j}_{i, \bu}(\ell) =\sum_{\substack{(\ell_1,\ell_2) \in \llbracket \ell - \phi(\bu) \rrbracket: \ell_1 + \ell_2 = \ell - \phi(\bu)}} \tilde{M}_{j,\bu[\key(j)]}(\ell_1)\cdot \tilde{W}^{\textsf{next}(j)}_{i, \bu}(\ell_2).
\end{equation}
Also, 
we define $\tilde{W}^{\textrm{next}(j^*)}_{i,\bu}(0) = 1$ and $\tilde{W}^{\textrm{next}(j^*)}_{i,\bu}(\ell) = 0$ for $\ell > 0$ for the largest node $R_{j^*} \in \C_i$.

We point out that $\tilde{W}^j_{i,\bu}(\ell)$ (resp. $\tilde{M}_{i,\bv}(\ell)$) is a constant-approximation of $W^j_{i,\bu}(\ell)$ (resp. $M_{i,\bv}(\ell)$). 
For each score $\ell \in \llbracket L-1\rrbracket$, and each child node $R_j \in \C_i$, we maintain a dynamic prefix-sum tree on the values $\tilde{W}^j_{i, \bu}(\ell)$ over all tuples $\bu \in R_i$ under the pre-determined ordering. This allows us to support point updates and prefix sum queries in $O(\log \In)$ time.

\smallskip
\noindent {\bf Update Procedure.}
When a new tuple $\bu$ is inserted into relation $R_i$, we first initialize its local statistics. If $R_i$ is a leaf node, $\tilde{W}^\emptyset_{i,\bu}(\ell) = 1$ if $\ell = \phi(\bu)$ and $\tilde{W}^\emptyset_{i,\bu}(\ell) = 0$ otherwise. If $R_i$ is an internal node, we compute $\tilde{W}^\emptyset_{i,\bu}(\ell)$ for all score $\ell \in \llbracket L-1\rrbracket$ via Equation~(\ref{eq:tilde-W}), which takes $O(L \log L)$ time using FFT. 
After initializing $\bu$, we must update the statistics in the parent node. Specifically, inserting $\bu$ increases $\hat{M}_{i,\bv}(\ell) = \sum_{\bu' \in R_i \ltimes \bv} \tilde{W}^\emptyset_{i,\bu'}(\ell)$ for $\bv = \bu[\key(i)]$. We update $\hat{M}_{i,\bv}(\ell)$ and check if $\tilde{M}_{i,\bv}(\ell)$ changes. If $\tilde{M}_{i,\bv}(\ell)$ changes, we must propagate this change to the parent relation $R_{\textsf{parent}(i)}$. This propagation is described in Algorithm~\ref{alg:update}. When a tuple $\bu$ is inserted into a leaf $R_i$, we just invoke \textsc{Update}$(i,\bu[\key(i)], \phi(\bu), 1)$. If it is inserted in an inner node we invoke \textsc{Update}$(i,\bu[\key(i)], \ell, \tilde{W}^\emptyset_{i,\bu}(\ell))$ for every $\ell\in \llbracket L-1\rrbracket$.

\smallskip
\noindent {\bf Complexity Analysis.} The update cost is dominated by the propagation of changes. Note that $\tilde{M}_{i,\bv}(\ell)$ is an upper bound that increases at most $O(\log \In)$ times for any fixed $i, \bv, \ell$, since $\tilde{M}_{i,\bv}(\ell)$ only increases if $\hat{M}_{i,\bv}(\ell)$ gets doubled.
When $\tilde{M}_{i,\bv}(\ell)$ changes, we perform convolution operations for tuples in the parent node that can be joined with $\bv$. Since $L = O(\log \In)$, the convolution takes $O(L \log L)$ time. 
We show that the total number of times $\tilde{M}$ values change across the entire stream of length $N$ is bounded as $O(NL\log N)$. So, the amortized update time is $O(L^2 \log \In \log L) = O(\log^3 \In \log\log \In)$. The space complexity remains $O(\In L) = O(\In \log \In)$.

\smallskip
\noindent {\bf Query Procedure.} The query procedure follows the optimized strategy from Section~\ref{sec:index-optimized}, invoking Algorithm~\ref{alg:ss-rejected-batch} on the score buckets $\mathcal{B}_\ell$ (for $\ell < L$) and the tail bucket $\mathcal{B}_{\ge L}$. To retrieve the $\tau$-th tuple from a selected bucket, we employ a modified version of Algorithm~\ref{alg:fast-recursive} that traverses the approximate statistics $\tilde{W}$ and $\tilde{M}$. As discussed in Section~\ref{sec:approximate-DA}, since these statistics are upper bounds, the index implicitly defines a superset of the join results containing ``dummy'' tuples. If the retrieved tuple is a dummy, we reject it. Since $\tilde{W}$ is a constant-factor approximation of $W$, the probability of drawing a real tuple is at least $2^{-c}$ for some constant $c$. Consequently, sampling $\chi$ tuples takes $O(\chi)$ expected time, and the overall expected query time remains $O(1 + \mu_\Psi \log \In)$.

We state the following theorem, with details deferred to Appendix~\ref{appndx:dynamic}.
\begin{theorem}\label{thm:dynamic}
    Given an acyclic join instance $Q$ and a set of associated weight functions $\{\fp_j\}_{R_j \in Q}$, we can maintain a dynamic index under tuple insertions using $O(\In \log \In)$ space, with amortized update time of $O(\log^3 \In \log\log\In)$, while supporting answering the subset sampling query in $O(1 +  \mu_\Psi \log \In)$ expected time, where $N$ is the total number of insertions. 
\end{theorem}

This index can be used to generate a one-shot subset sample with the following running time:
\vspace{-0.1em}
\begin{corollary}
    \label{cor:dynamic}
    Given an acyclic join instance $Q$ and a set of associated weight functions $\{\fp_j\}_{R_j \in Q}$, we can maintain a one-shot subset sample in $O(N \log^3 \In \log\log\In +  \mu_\Psi \log \In)$ expected time, where $N$ is the total number of insertions. 
\end{corollary}

\vspace{-0.5em}
\section{Conclusion}
\label{sec:conclusion}
\vspace{-0.2em}
 
We presented the first efficient framework for subset sampling over joins without full materialization. By exploiting the join structure of query and decomposability of weight functions, our algorithms achieve near-optimal complexity for static indexing, one-shot sampling and dynamic maintenance under insertions. This work opens several intriguing avenues for future research. We aim to extend our theoretical framework to support a broader class of weight definitions. While this paper addressed standard aggregation functions, efficient sampling under complex, non-monotonic, or holistic aggregation functions remains open. Furthermore, we plan to bridge the gap between theory and practice by implementing and evaluating our indexes on real-world machine learning workloads over relational data. 

\newpage
\bibliographystyle{abbrvnat}
\bibliography{reference}

\begin{thebibliography}{46}
\providecommand{\natexlab}[1]{#1}
\providecommand{\url}[1]{\texttt{#1}}
\expandafter\ifx\csname urlstyle\endcsname\relax
  \providecommand{\doi}[1]{doi: #1}\else
  \providecommand{\doi}{doi: \begingroup \urlstyle{rm}\Url}\fi

\bibitem[Abiteboul et~al.(1995)Abiteboul, Hull, and Vianu]{abiteboul1995foundations}
S.~Abiteboul, R.~Hull, and V.~Vianu.
\newblock \emph{Foundations of databases}, volume~8.
\newblock Addison-Wesley Reading, 1995.

\bibitem[Abo-Khamis et~al.(2021)Abo-Khamis, Im, Moseley, Pruhs, and Samadian]{abo2021relational}
M.~Abo-Khamis, S.~Im, B.~Moseley, K.~Pruhs, and A.~Samadian.
\newblock A relational gradient descent algorithm for support vector machine training.
\newblock In \emph{Symposium on Algorithmic Principles of Computer Systems (APOCS)}, pages 100--113. SIAM, 2021.

\bibitem[Agarwal et~al.(2024)Agarwal, Esmailpour, Hu, Sintos, and Yang]{agarwal2024computing}
P.~K. Agarwal, A.~Esmailpour, X.~Hu, S.~Sintos, and J.~Yang.
\newblock Computing a well-representative summary of conjunctive query results.
\newblock \emph{Proceedings of the ACM on Management of Data}, 2\penalty0 (5):\penalty0 1--27, 2024.

\bibitem[Arenas et~al.(2024)Arenas, Merkl, Pichler, and Riveros]{arenas2024towards}
M.~Arenas, T.~C. Merkl, R.~Pichler, and C.~Riveros.
\newblock Towards tractability of the diversity of query answers: Ultrametrics to the rescue.
\newblock \emph{Proceedings of the ACM on Management of Data}, 2\penalty0 (5):\penalty0 1--26, 2024.

\bibitem[Atserias et~al.(2013)Atserias, Grohe, and Marx]{atserias2013size}
A.~Atserias, M.~Grohe, and D.~Marx.
\newblock Size {Bounds} and {Query} {Plans} for {Relational} {Joins}.
\newblock \emph{SIAM Journal of Computing}, 42\penalty0 (4):\penalty0 1737--1767, 2013.

\bibitem[Bachem et~al.(2017)Bachem, Lucic, and Krause]{bachem2017practical}
O.~Bachem, M.~Lucic, and A.~Krause.
\newblock Practical coreset constructions for machine learning.
\newblock \emph{arXiv preprint arXiv:1703.06476}, 2017.

\bibitem[Berkholz et~al.(2017)Berkholz, Keppeler, and Schweikardt]{berkholz17:_answer}
C.~Berkholz, J.~Keppeler, and N.~Schweikardt.
\newblock Answering conjunctive queries under updates.
\newblock In \emph{Proceedings of the 36th ACM SIGMOD-SIGACT-SIGAI Symposium on Principles of Database Systems}, PODS '17, page 303–318, New York, NY, USA, 2017. Association for Computing Machinery.
\newblock ISBN 9781450341981.

\bibitem[Bhattacharya et~al.(2024)Bhattacharya, Kiss, Sidford, and Wajc]{bhattacharya2024nearoptimal}
S.~Bhattacharya, P.~Kiss, A.~Sidford, and D.~Wajc.
\newblock Near-{Optimal} {Dynamic} {Rounding} of {Fractional} {Matchings} in {Bipartite} {Graphs}.
\newblock In \emph{Proceedings of {ACM} Symposium on Theory of Computing ({STOC})}, pages 59--70. ACM, 2024.

\bibitem[Borodin and Munro(1975)]{BM75}
A.~Borodin and J.~I. Munro.
\newblock \emph{The computational complexity of algebraic and numeric problems}, volume~1.
\newblock Elsevier, 1975.

\bibitem[Bringmann and Panagiotou(2012)]{bringmann2012efficient}
K.~Bringmann and K.~Panagiotou.
\newblock Efficient sampling methods for discrete distributions.
\newblock In \emph{Proceedings of International Colloquium on Automata, Languages and Programming ({ICALP})}, pages 133--144, 2012.

\bibitem[Carmeli et~al.(2020)Carmeli, Zeevi, Berkholz, Kimelfeld, and Schweikardt]{carmeli2020answering}
N.~Carmeli, S.~Zeevi, C.~Berkholz, B.~Kimelfeld, and N.~Schweikardt.
\newblock Answering (unions of) conjunctive queries using random access and random-order enumeration.
\newblock In \emph{Proceedings of the 39th ACM SIGMOD-SIGACT-SIGAI Symposium on Principles of Database Systems}, pages 393--409, 2020.

\bibitem[Carmeli et~al.(2023)Carmeli, Tziavelis, Gatterbauer, Kimelfeld, and Riedewald]{carmeli2023tractable}
N.~Carmeli, N.~Tziavelis, W.~Gatterbauer, B.~Kimelfeld, and M.~Riedewald.
\newblock Tractable orders for direct access to ranked answers of conjunctive queries.
\newblock \emph{ACM Transactions on Database Systems}, 48\penalty0 (1):\penalty0 1--45, 2023.

\bibitem[Chaudhuri et~al.(1999)Chaudhuri, Motwani, and Narasayya]{chaudhuri1999random}
S.~Chaudhuri, R.~Motwani, and V.~Narasayya.
\newblock On random sampling over joins.
\newblock In \emph{Proceedings of the 1999 ACM SIGMOD International Conference on Management of Data}, SIGMOD '99, page 263–274, New York, NY, USA, 1999. Association for Computing Machinery.
\newblock ISBN 1581130848.

\bibitem[Chen et~al.(2022)Chen, Yang, Huang, and Ding]{chen2022coresets}
J.~Chen, Q.~Yang, R.~Huang, and H.~Ding.
\newblock Coresets for relational data and the applications.
\newblock \emph{Advances in Neural Information Processing Systems}, 35:\penalty0 434--448, 2022.

\bibitem[Chen and Yi(2020)]{chen2020random}
Y.~Chen and K.~Yi.
\newblock {Random Sampling and Size Estimation Over Cyclic Joins}.
\newblock In C.~Lutz and J.~C. Jung, editors, \emph{23rd International Conference on Database Theory (ICDT 2020)}, volume 155 of \emph{Leibniz International Proceedings in Informatics (LIPIcs)}, pages 7:1--7:18, Dagstuhl, Germany, 2020. Schloss Dagstuhl -- Leibniz-Zentrum f{\"u}r Informatik.
\newblock ISBN 978-3-95977-139-9.

\bibitem[Cheng and Koudas(2019)]{cheng2019nonlinear}
Z.~Cheng and N.~Koudas.
\newblock Nonlinear models over normalized data.
\newblock In \emph{2019 IEEE 35th International Conference on Data Engineering (ICDE)}, pages 1574--1577. IEEE, 2019.

\bibitem[Cheng et~al.(2021)Cheng, Koudas, Zhang, and Yu]{cheng2021efficient}
Z.~Cheng, N.~Koudas, Z.~Zhang, and X.~Yu.
\newblock Efficient construction of nonlinear models over normalized data.
\newblock In \emph{2021 IEEE 37th International Conference on Data Engineering (ICDE)}, pages 1140--1151. IEEE, 2021.

\bibitem[Cooley and Tukey(1965)]{cooley1965algorithm}
J.~W. Cooley and J.~W. Tukey.
\newblock An algorithm for the machine calculation of complex fourier series.
\newblock \emph{Mathematics of computation}, 19\penalty0 (90):\penalty0 297--301, 1965.

\bibitem[Cormen et~al.(2022)Cormen, Leiserson, Rivest, and Stein]{cormen2022introduction}
T.~H. Cormen, C.~E. Leiserson, R.~L. Rivest, and C.~Stein.
\newblock \emph{Introduction to algorithms}.
\newblock MIT press, 2022.

\bibitem[Cormode et~al.(2012)Cormode, Garofalakis, Haas, and Jermaine]{cormode2012synopses}
G.~Cormode, M.~Garofalakis, P.~J. Haas, and C.~Jermaine.
\newblock Synopses for massive data: Samples, histograms, wavelets, sketches.
\newblock \emph{Foundations and Trends{\textregistered} in Databases}, 4\penalty0 (1--3):\penalty0 1--294, 2012.

\bibitem[Curtin et~al.(2020)Curtin, Moseley, Ngo, Nguyen, Olteanu, and Schleich]{curtin2020rk}
R.~Curtin, B.~Moseley, H.~Ngo, X.~Nguyen, D.~Olteanu, and M.~Schleich.
\newblock Rk-means: Fast clustering for relational data.
\newblock In \emph{International Conference on Artificial Intelligence and Statistics}, pages 2742--2752. PMLR, 2020.

\bibitem[Dai et~al.(2024)Dai, Hu, and Yi]{DHY24}
B.~Dai, X.~Hu, and K.~Yi.
\newblock Reservoir sampling over joins.
\newblock In \emph{Proceedings of ACM Management of Data ({SIGMOD})}, pages 1--26, 2024.

\bibitem[Deng et~al.(2023)Deng, Lu, and Tao]{deng2023join}
S.~Deng, S.~Lu, and Y.~Tao.
\newblock On join sampling and the hardness of combinatorial output-sensitive join algorithms.
\newblock In \emph{Proceedings of the 42nd ACM SIGMOD-SIGACT-SIGAI Symposium on Principles of Database Systems}, PODS '23, page 99–111, New York, NY, USA, 2023. Association for Computing Machinery.
\newblock ISBN 9798400701276.

\bibitem[Duffield et~al.(2007)Duffield, Lund, and Thorup]{duffield2007priority}
N.~Duffield, C.~Lund, and M.~Thorup.
\newblock Priority sampling for estimation of arbitrary subset sums.
\newblock \emph{Journal of the ACM ({JACM})}, 54\penalty0 (6):\penalty0 32, 2007.

\bibitem[Esmailpour and Sintos(2024)]{esmailpour2024improved}
A.~Esmailpour and S.~Sintos.
\newblock Improved approximation algorithms for relational clustering.
\newblock \emph{Proceedings of the ACM on Management of Data}, 2\penalty0 (5):\penalty0 1--27, 2024.

\bibitem[Feldman and Langberg(2011)]{feldman2011unified}
D.~Feldman and M.~Langberg.
\newblock A unified framework for approximating and clustering data.
\newblock In \emph{Proceedings of {ACM} Symposium on Theory of Computing ({STOC})}, pages 569--578, 2011.

\bibitem[Gottlob et~al.(2014)Gottlob, Greco, Scarcello, et~al.]{gottlob2014treewidth}
G.~Gottlob, G.~Greco, F.~Scarcello, et~al.
\newblock Treewidth and hypertree width.
\newblock \emph{Tractability: Practical Approaches to Hard Problems}, 1:\penalty0 20, 2014.

\bibitem[Haas and Hellerstein(1999)]{haas1999ripple}
P.~J. Haas and J.~M. Hellerstein.
\newblock Ripple joins for online aggregation.
\newblock In \emph{Proceedings of the 1999 ACM SIGMOD International Conference on Management of Data}, SIGMOD '99, page 287–298, New York, NY, USA, 1999. Association for Computing Machinery.
\newblock ISBN 1581130848.

\bibitem[Huang and Wang(2023)]{huang2023subset}
J.~Huang and S.~Wang.
\newblock Subset {Sampling} and {Its} {Extensions}.
\newblock \emph{arXiv preprint arXiv:2003.01075}, 2023.

\bibitem[Kara et~al.(2024)Kara, Nikolic, Olteanu, and Zhang]{kara2024f}
A.~Kara, M.~Nikolic, D.~Olteanu, and H.~Zhang.
\newblock F-ivm: analytics over relational databases under updates.
\newblock \emph{The VLDB Journal}, 33\penalty0 (4):\penalty0 903--929, 2024.

\bibitem[Khamis et~al.(2018)Khamis, Ngo, Nguyen, Olteanu, and Schleich]{khamis2018ac}
M.~A. Khamis, H.~Q. Ngo, X.~Nguyen, D.~Olteanu, and M.~Schleich.
\newblock Ac/dc: in-database learning thunderstruck.
\newblock In \emph{Proceedings of the second workshop on data management for end-to-end machine learning}, pages 1--10, 2018.

\bibitem[Kim et~al.(2023)Kim, Ha, Fletcher, and Han]{Kim2023guaranteeing}
K.~Kim, J.~Ha, G.~Fletcher, and W.-S. Han.
\newblock Guaranteeing the \~{O}(agm/out) runtime for uniform sampling and size estimation over joins.
\newblock In \emph{Proceedings of the 42nd ACM SIGMOD-SIGACT-SIGAI Symposium on Principles of Database Systems}, PODS '23, page 113–125, New York, NY, USA, 2023. Association for Computing Machinery.
\newblock ISBN 9798400701276.

\bibitem[Kumar et~al.(2015)Kumar, Naughton, and Patel]{kumar2015learning}
A.~Kumar, J.~Naughton, and J.~M. Patel.
\newblock Learning generalized linear models over normalized data.
\newblock In \emph{Proceedings of ACM Management of Data ({SIGMOD})}, pages 1969--1984, 2015.

\bibitem[Li et~al.(2016)Li, Wu, Yi, and Zhao]{li2016wander}
F.~Li, B.~Wu, K.~Yi, and Z.~Zhao.
\newblock Wander join: Online aggregation via random walks.
\newblock In \emph{Proceedings of the 2016 International Conference on Management of Data}, pages 615--629, 2016.

\bibitem[Merkl et~al.(2025)Merkl, Pichler, and Skritek]{merkl2025diversity}
T.~C. Merkl, R.~Pichler, and S.~Skritek.
\newblock Diversity of answers to conjunctive queries.
\newblock \emph{Logical Methods in Computer Science}, 21, 2025.

\bibitem[Moseley et~al.(2021)Moseley, Pruhs, Samadian, and Wang]{moseley2021relational}
B.~Moseley, K.~Pruhs, A.~Samadian, and Y.~Wang.
\newblock Relational algorithms for k-means clustering.
\newblock In \emph{International Colloquium on Automata, Languages, and Programming}, 2021.

\bibitem[Motwani and Raghavan(1995)]{motwani1995randomized}
R.~Motwani and P.~Raghavan.
\newblock \emph{Randomized algorithms}.
\newblock Cambridge University Press, 1995.

\bibitem[Ngo(2018)]{ngo2018worst}
H.~Q. Ngo.
\newblock Worst-case optimal join algorithms: Techniques, results, and open problems.
\newblock In \emph{Proceedings of ACM Symposium on Principles of Database Systems ({PODS})}, pages 111--124. ACM, 2018.

\bibitem[Olken(1993)]{olken1993random}
F.~Olken.
\newblock \emph{Random sampling from databases}.
\newblock PhD thesis, University of California, Berkeley, 1993.

\bibitem[Preparata and Shamos(1985)]{PS85}
F.~P. Preparata and M.~I. Shamos.
\newblock Computational geometry. texts and monographs in computer science.
\newblock \emph{Berlin, Springer-Verlag}, 1985.

\bibitem[Schleich et~al.(2016)Schleich, Olteanu, and Ciucanu]{schleich2016learning}
M.~Schleich, D.~Olteanu, and R.~Ciucanu.
\newblock Learning linear regression models over factorized joins.
\newblock In \emph{Proceedings of ACM Management of Data ({SIGMOD})}, pages 3--18, 2016.

\bibitem[Schleich et~al.(2019)Schleich, Olteanu, Abo-Khamis, Ngo, and Nguyen]{schleich2019learning}
M.~Schleich, D.~Olteanu, M.~Abo-Khamis, H.~Q. Ngo, and X.~Nguyen.
\newblock Learning models over relational data: A brief tutorial.
\newblock In \emph{Scalable Uncertainty Management: 13th International Conference, SUM 2019}, pages 423--432. Springer, 2019.

\bibitem[Yang et~al.(2020)Yang, Gao, Liang, Yao, Wen, and Chen]{yang2020towards}
K.~Yang, Y.~Gao, L.~Liang, B.~Yao, S.~Wen, and G.~Chen.
\newblock Towards factorized svm with gaussian kernels over normalized data.
\newblock In \emph{2020 IEEE 36th International Conference on Data Engineering (ICDE)}, pages 1453--1464. IEEE, 2020.

\bibitem[Yannakakis(1981)]{yannakakis1981algorithms}
M.~Yannakakis.
\newblock Algorithms for acyclic database schemes.
\newblock In \emph{Proceedings of the Seventh International Conference on Very Large Data Bases - Volume 7}, VLDB '81, page 82–94. VLDB Endowment, 1981.

\bibitem[Yi et~al.(2023)Yi, Wang, and Wei]{yi2023optimal}
L.~Yi, H.~Wang, and Z.~Wei.
\newblock Optimal {Dynamic} {Subset} {Sampling}: Theory and {Applications}.
\newblock In \emph{Proceedings of ACM Knowledge Discovery and Data Mining ({SIGKDD})}, pages 3116--3127. ACM, 2023.

\bibitem[Zhao et~al.(2018)Zhao, Christensen, Li, Hu, and Yi]{zhao2018random}
Z.~Zhao, R.~Christensen, F.~Li, X.~Hu, and K.~Yi.
\newblock Random sampling over joins revisited.
\newblock In \emph{Proceedings of ACM Management of Data ({SIGMOD})}, 2018.

\end{thebibliography}

\appendix

\section{Full Notation Table}\label{appendix:notations}
\begin{table}[H]
\centering
\renewcommand{\arraystretch}{1.1}
\begin{tabular}{c|c|l}
\toprule & \textbf{Notation} & \textbf{Description} \\
\midrule
\midrule
\multirow{7}{*}{\rotatebox[origin=c]{90}{\shortstack[c]{\textbf{General}}}} 
  & $S$ & A set of elements. \\ 
  & $\fp$ & A probability function assigning a probability to each element in $S$. \\ 
  & $\Psi = \langle S, \fp \rangle$ & A subset sampling problem instance. \\ 
  & $\mathbf{X}$ & A random subset sample from a subset sampling instance. \\ 
  & $\mu_\Psi$ & The expected size of a subset sample $\mathbf{X}$ from $\Psi$. \\ 
  & $[x]$ & The set of integers $\{1, 2, \ldots, x\}$. \\ 
  & $\llbracket x \rrbracket$ & The set of integers $\{0,1, 2, \ldots, x\}$. \\ 
  & \DA & A Direct Access oracle. \\
\midrule
\midrule
\multirow{14}{*}{\rotatebox[origin=c]{90}{\shortstack[c]{\textbf{Subset Sampling over Joins \ \ \ \ \ \ \ \ }}}} 
  & $\att, \dom$ & The set of all attributes and the domain of all values. \\ 
  & $\bu, \bv$ & Tuples. \\ 
  & $R_i$ & The $i$-th relation in a join query. \\ 
  & $\schema(R_i)$ & The schema (set of attributes) of relation $R_i$. \\ 
  & $Q = \{R_1, \ldots, R_k\}$ & A join query consisting of $k$ relations. \\ 
  & $\In$ & The input size of $Q$, i.e., $\sum_{i=1}^k |R_i|$. \\ 
  & $\join(Q)$ & The set of result tuples from the join query $Q$. \\ 
  & $G=(V,E)$ & The schema hypergraph of a join query $Q$. \\ 
  & $\rho^*$ & The fractional edge covering number of schema graph $G$. \\
  & $\fp_i(\bu)$ & The probability associated with a tuple $\bu \in R_i$. \\ 
  & $\fp(\bu)$ & Prob. of join result $\bu$, defined as $\prod_{i=1}^k \fp_i(\bu[\schema(R_i)])$. \\ 
  & $L$ & The number of partitions/buckets, defined as $\lceil\rho^*\log\In\rceil$. \\ 
  & $\phi(\bu)$ & The score of a tuple $\bu$. \\ 
  & $\mathcal{B}_\ell$ & The set (bucket) of join results with score $\ell$. \\ 
\midrule
\midrule
\multirow{8}{*}{\rotatebox[origin=c]{90}{\shortstack[c]{\textbf{Optimized Index}}}} 
  & $\mathcal{T}$ & Join tree for an acyclic join. \\ 
  & $\T_i$ & The subtree of $\T$ rooted at relation $R_i$. \\
  & $\C_i$ & The set of child nodes of relation $R_i$ in $\T$. \\
  & $\T_i^j$ & The partial subtree of $\T_i$ excluding children preceding $R_j$. \\
  & $\key(i)$ & Common join attributes between $R_i$ and its parent in $\mathcal{T}$. \\ 
  & $W^j_{i,\bu}(\ell)$ & \# of join results with score $\ell$ in $\T_i^j$ involving $\bu$. \\ 
  & $M_{i,\bv}(\ell)$ & An aggregated count of $W$-values for tuples projecting to $\bv$. \\
\midrule
\midrule
\multirow{5}{*}{\rotatebox[origin=c]{90}{\shortstack[c]{\textbf{Dynamic}}}} 
  & $\eta$ & A timestamp in the streaming setting. \\
  & $Q^\eta$ & The join defined by the first $\eta$ tuples in the stream. \\
  & $\fp^\eta$ & The weight functions associated with $Q^\eta$. \\
  & $\tilde{W}^j_{i,\bu}(\ell)$ & An approximated upper bound of $W^j_{i,\bu}(\ell)$. \\
  & $\tilde{M}_{i,\bv}(\ell)$ & An approximated upper bound of $M_{i,\bv}(\ell)$. \\
\bottomrule
\end{tabular}
\caption{Table of Notations} \label{tab:notations-full}
\end{table}

\section{Missing Materials in Section \ref{sec:revisited}}
\label{appendix:revisited}

\smallskip
\begin{proof}[Proof of Lemma~\ref{lm:batched-index}]
The preprocessing phase constructs two levels of indexes: standard subset sampling indexes for each sub-instance $\Psi_i$ (taking $O(\sum |S_i|) = O(|S|)$ time) and a meta-index for $\Psi_{\text{meta}} = \langle [m], \fq \rangle$ (taking $O(m) \le O(|S|)$ time). The space complexity is clearly $O(|S|)$.

We now analyze the expected query time of Algorithm~\ref{alg:ss-rejected-batch}. The running time consists of two parts:
    
\noindent \textbf{Meta-sampling:} Using the meta-index, generating the set of indices $\rdI$ takes expected time $O(1 + \mu_{\Psi_{\text{meta}}})$, where $\mu_{\Psi_{\text{meta}}} = \sum_{i=1}^m \fq(i)$.
    
\noindent \textbf{Sub-sampling:} For each selected index $i \in \rdI$, we generate an intermediate sample $\rdX'_i$ using the geometric jump procedure (simulating \textbf{uss-advanced}). The cost for a specific $i$ is proportional to the size of this intermediate sample. The total expected cost is:
\[ 
    \expt[\text{Cost}_{\text{sub}}] = \sum_{i=1}^m \fq(i) \cdot O\left( \expt[|\rdX'_i| \mid \rdX'_i \neq \emptyset] \right). 
\]
Note that $\expt[|\rdX'_i| \mid \rdX'_i \neq \emptyset] = \frac{\expt[|\rdX'_i|]}{\Pr[\rdX'_i \neq \emptyset]} = \frac{|S_i|p_i^+}{\fq(i)}$. Substituting this back into the summation, the $\fq(i)$ terms cancel out:
\[ 
    \expt[\text{Cost}_{\text{sub}}] = O\left( \sum_{i=1}^m \fq(i) \cdot \frac{|S_i|p_i^+}{\fq(i)} \right) = O\left( \sum_{i=1}^m |S_i|p_i^+ \right). 
\]
    
Combining both parts, the total expected time is $O(1 + \sum_{i=1}^m \fq(i) + \sum_{i=1}^m |S_i|p_i^+)$. Since $\fq(i) \le |S_i|p_i^+$, the total complexity simplifies to $O(1 + \sum_{i=1}^m |S_i|p_i^+)$.

Finally, if each $\Psi_i$ is $\beta$-uniform, we have $p_i^+ \le \beta \cdot \min_{e \in S_i} \fp_i(e)$. Summing over all elements in $S_i$, we get $|S_i| p_i^+ \le \beta \sum_{e \in S_i} \fp_i(e) = \beta \mu_{\Psi_i}$. Summing over all $i$, the total cost is bounded by $O(1 + \sum_i \beta \mu_{\Psi_i}) = O(1 + \beta \mu_\Psi)$.
\end{proof}

\section{Missing Materials in Section~\ref{sec:static-index}}
\label{appendix:static-index}
\subsection{Analysis of the High-level Framework}
\begin{proof}[Proof of Lemma~\ref{lem:uniform-light}]
     For the largest probability, we always have $\displaystyle{\max_{\bu\in\join(Q_{\bj})}\fp(\bu)\le2^{-\sum_{i\in[k]} j_i}}$.
     If $\sum_{i\in[k]} j_i \ge L$, we have $\max_{\bu\in\join(Q_{\bj})}\fp(\bu) \le 2^{-\sum_{i\in[k]} j_i} \le 2^{-L} \le 1/|\join(Q)|$, following our assumption on $L$. Hence, $\Psi_{\bj}$ is light. Otherwise, we must have $L\notin\{j_1,j_2,\ldots,j_k\}$. In this case, for the smallest probability, we have $\displaystyle{2^{-k-\sum_{i\in[k]} j_i} \le  \min_{\bu\in\join(Q_{\bj})}\fp(\bu)}$. Hence, $\Psi_{\bj}$ is $2^k$-uniform since $\max_{\bu\in\join(Q_{\bj})}\fp(\bu) \le  \min_{\bu\in\join(Q_{\bj})}\fp(\bu)$
\end{proof}

\begin{proof}[Proof of Theorem~\ref{thm:logk}]
The proof follows from analyzing the preprocessing cost and the query time of our sampling strategy.

\extraspacing{\bf Preprocessing and Data Structure.} The preprocessing phase constructs the composite index defined in Lemma~\ref{lm:batched-index}.
\begin{enumerate}[leftmargin=*]
    \item {\bf Partitioning and Sub-indexes.} We partition each relation $R_i \in Q$ into $L+1$ sub-relations, defining $(L+1)^k$ sub-instances $Q_{\bj}$. For each sub-instance, we construct a \DA oracle (sub-index). According to~\cite{zhao2018random}, building a \DA oracle for an acyclic join takes time linear in its input size. Since each tuple participates in $(L+1)^{k-1}$ sub-instances and $L = O(\log \In)$, the total time and space are $O(\In \log^{k-1}\In)$.

    \item {\bf Meta-index.} We build the meta-index for $\Psi' = \langle \llbracket L \rrbracket^k, \fq \rangle$. As the universe size is $(L+1)^k = O(\log^k \In)$, this index can be built in $O(\log^k \In)$ time and space, which is subsumed by the cost of building the sub-indexes.
\end{enumerate}

\extraspacing{\bf Query Answering.} The query proceeds according to Algorithm~\ref{alg:ss-rejected-batch}.
\begin{enumerate}[leftmargin=*]
    \item {\bf Stage 1: Meta-sampling.} We query the meta-index to obtain indices $\mathbf{I}$. The expected time is $O(1 + \mu_{\Psi'})$, where $\mu_{\Psi'} = \sum_{\bj} \fq(\bj)$. As shown in the main text, $\mu_{\Psi'} \le \mu_\Psi + 1$.
    
    \item {\bf Stage 2: Sub-sampling.} For each $\bj \in \mathbf{I}$, we draw samples from $Q_{\bj}$. Since each sub-instance is either light or $2^k$-uniform, the expected cost per selected instance is proportional to the number of samples times the \DA access cost $O(\log N)$. Following the analysis in Lemma~\ref{lm:batched-index}, the total expected cost is $O((1 + \mu_\Psi) \log \In)$.
\end{enumerate}
Summing these gives the total expected query time $O(1 + \mu_\Psi \log \In)$.
\end{proof}

 \subsection{Preprocessing for Optimized Algorithm}
 We will show how to compute the statistics to build the desired index. The complete pseudocode is provided in Algorithm~\ref{alg:preprocess}. Below, we show the intuition in detail. Our preprocessing phase performs computation in a bottom-up way. We note that $M$-values are defined on top of $W$-values, hence for each node $R_i \in Q$, $M$-values will be computed according to (\ref{eq:M}) straightforwardly once $W$-values are well computed. Below, we focus on computing $W$-values.
   
   Consider a leaf node $R_i$. By definition of (\ref{eq:W}), we have $W^\emptyset_{i,\bu}(\phi(\bu))=1$ and $W^\emptyset_{i,\bu}(l)=0$ for any $l \in \llbracket L-1 \rrbracket\setminus\{\phi(\bu)\}$. Consider an internal node $R_i$. Suppose $W$-values and $M$-values are well defined for each child $R_j \in \mathcal{C}_i$. We next compute $W$-values for each tuple $\bu \in R_i$ and each child node $R_j \in \C_i$. More specifically, we compute $W^{j}_{i, \bu}$ in a decreasing ordering of nodes in $\C_i$. Recall that for the largest node $R_{j^*} \in \C_i$, $W^{j^*}_{i, \bu}(0)=1$ and $W^{j^*}_{i, \bu}(\ell)=0$ for $\ell \in [L-1]$. We next compute $W^{j}_{i, \bu}(\cdot)$, assuming $W^{\textsf{next}(j)}_{i,\bu}(\cdot)$ has been computed. 
   We distinguish the following two cases on the score $\ell$ (recall that $\ell \in \llbracket L-1\rrbracket$): 
\begin{itemize}[leftmargin=*]
    \item {\bf Case 1: $\ell < \phi(\bu)$}. In this case, $W^j_{i,\bu}(\ell) =0$ since any join result participated by $\bu$ has its score at least $\phi(\bu)$;
    \item {\bf Case 2: $\phi(\bu) \le \ell < L$.} First, the Cartesian product of join result in each subtree $\T_{j'}$ rooted at any child node $R_{j'} \in \C_i$ coming no earlier than $R_j$, that can be joined with $\bu$, is essentially the join result in $\T^j_i$ that can be joined with $\bu$. We also need to take the scores into consideration---all possible combinations of $\langle \ell_{j'}\in \llbracket L-1 \rrbracket: R_{j'} \in \mathcal{C}_i, R_j \preceq R_{j'}\rangle$ with $\sum_{R_{j'} \in \mathcal{C}_i} \ell_{j'} = \ell - \phi(\bu)$. Each combination will contribute $\prod_{R_{j'} \in \mathcal{C}_i: R_j \preceq R_{j'}} M_{j',\bu[\key(j')]}(\ell_{j'})$ to $W^j_{i,\bu}(\ell)$. However, the number of such combinations is as large as $O(L^k)$, which translates to $O(\log^{k} \In)$ time even for computing a single $W$-value. A critical observation is that the products of $M$-values over $R_{j'}$ with $R_j \prec R_{j'}$ is actually captured by $W^{\textsf{next}(j)}_{i, \bu}$, where $R_{\textsf{next}(j)}$ is the immediately next node of $R_j$ in $\C_i$. Putting these observations together, we obtain:
    \begin{align}
    \label{W-recursive-apd}
        W^j_{i,\bu}(\ell) = & 
         \displaystyle{\sum_{\substack{\langle \ell_{j'}\in \llbracket L - 1\rrbracket: R_{j'} \in \mathcal{C}_i, R_j \preceq R_{j'}\rangle: \\
         \ell - \phi(\bu) = \sum_{R_{j'} \in \mathcal{C}_i} \ell_{j'}}} \ \prod_{R_{j'} \in \mathcal{C}_i: R_j \preceq R_{j'}} M_{j',\bu[\key(j')]}(\ell_{j'})} \nonumber \\
        = & \sum_{\substack{(\ell_1,\ell_2) \in \llbracket \ell - \phi(\bu) \rrbracket: \\ \ell_1 + \ell_2 = \ell - \phi(\bu)}} M_{j,\bu[\key(j)]}(\ell_1)\cdot  \displaystyle{\sum_{\substack{\langle \ell_{j'}\in \llbracket L-1 \rrbracket: R_{j'} \in \mathcal{C}_i, R_j \prec R_{j'}\rangle: \\ \ell_2 = \sum_{R_j \in \mathcal{C}_i} \ell_j}} \ \prod_{R_{j'} \in \mathcal{C}_i: R_j \prec R_{j'}} M_{j',\bu[\key(j')]}(\ell_{j'})} \nonumber \\
        =& \sum_{\substack{(\ell_1,\ell_2) \in \llbracket \ell - \phi(\bu) \rrbracket: \\ \ell_1 + \ell_2 = \ell - \phi(\bu)}} M_{j,\bu[\key(j)]}(\ell_1)\cdot W^{\textsf{next}(j)}_{i, \bu}(\ell_2).
    \end{align}
\end{itemize}
Hence $W^j_{i,\bu}(\ell)$ can computed by (\ref{W-recursive-apd}). At that point, the $M$-values for the child node $R_j$ and the $W$-values for the subsequent node $R_{\textsf{next}(j)}$ have already been computed. 
Note that $W^\emptyset_{i,\bu}(\ell)$ can be computed similarly by (\ref{W-recursive-apd}):
$$W^\emptyset_{i,\bu}(\ell) = \sum_{\substack{(\ell_1,\ell_2) \in \llbracket \ell - \phi(\bu) \rrbracket: \\ \ell_1 + \ell_2 = \ell - \phi(\bu)}} M_{j^\circ,\bu[\key(j^\circ)]}(\ell_1)\cdot W^{\textsf{next}(j^\circ)}_{i, \bu}(\ell_2),$$
where $R_{j^\circ}$ is the first node in $\C_i$.

To speed up the process, let's fix a node $R_i \in Q$, a tuple $\bu \in R_i$, and a child node $R_j \in \C_i$ as an example. Now, we need to compute $W^\emptyset_{i,\bu}(\cdot)$ for $L$ different scores. There is a convolution structure between $W^{j}_{i, \bu}(\cdot)$, $M_{j, \key(j)}(\cdot)$ and $W^{\textsf{next}(j)}_{i, \bu}(\cdot)$, so we apply the Fast Fourier Transform for computing these $L$ scores together. As we will show in the analysis of the preprocessing step, this observation improves the preprocessing time from $O(\In L^2)$ to $O(\In L\log L)$ (Lemma~\ref{lem:preprocessing}). 

After all $W$-values are computed, for each internal node $R_i\in Q$ with each child node $R_j \in \C_i$, each value $\bv \in R_i[\key(i)]$, and each score $\ell \in \llbracket L-1\rrbracket$, we build a prefix-sum array on $\langle W^j_{i, \bu}(\ell): \bu \in R_i\ltimes \bv\rangle$ under the pre-determined ordering. 

Lastly, we compute the bucket size $|\B_\ell|$ for each score $\ell \in \llbracket L-1\rrbracket$, which by definition can be rewritten as $|\B_\ell| = \sum_{\bu \in R_r} W^\emptyset_{r,\bu}(\ell)$ for the root node $r$.

\begin{algorithm}[h]
    \caption{\textbf{Preprocess$(Q)$}}
    \label{alg:preprocess}
    \KwIn{Join tree $\T$, Relations $\{R_i\}_{R_i \in \T}$.}
    \ForEach{node $R_i \in \T$ in bottom-up order}{
        \If{$R_i$ is a leaf node}{
            \ForEach{tuple $\bu \in R_i$}{
                Set $W^\emptyset_{i,\bu}(\phi(\bu)) \gets 1$ and $W^\emptyset_{i,\bu}(\ell) \gets 0$ for $\ell \neq \phi(\bu)$\;
                $M_{i, \bu[\key(i)]} \gets M_{i, \bu[\key(i)]} + W^\emptyset_{i,\bu}$\tcp*[r]{element-wise vector addition}
            }
        }
        \Else(\tcp*[h]{$R_i$ is an internal node}){
            \ForEach{tuple $\bu \in R_i$}{
                Let $R_{j^*}$ be the last child in $\C_i$\;
                Set $W^{\textsf{next}(j^*)}_{i,\bu}(0) \gets 1$ and $W^{\textsf{next}(j^*)}_{i,\bu}(\ell) \gets 0$ for $\ell > 0$ \;
                \ForEach{child $R_j \in \C_i$ in decreasing order}{
                    Let $R_{\textsf{next}(j)}$ be the child immediately following $R_j$\;
                    \ForEach{score $\ell \in \llbracket L-1\rrbracket$}{
                        \If{$\ell < \phi(\bu)$}{
                            $W^j_{i,\bu}(\ell) \gets 0$\;
                        }
                        \Else{
                            $W^j_{i,\bu}(\ell) \gets \sum_{k=0}^{\ell - \phi(\bu)} M_{j, \bu[\key(j)]}(k) \cdot W^{\textsf{next}(j)}_{i, \bu}(\ell - \phi(\bu) - k)$\;
                        }
                    }
                }
                $M_{i, \bu[\key(i)]} \gets M_{i, \bu[\key(i)]} + W^\emptyset_{i,\bu}$\tcp*[r]{element-wise vector addition}
            }
            \ForEach{child $R_j \in \C_i$}{
                \ForEach{$\bv \in R_i[\key(i)]$ and $\ell \in \llbracket L-1\rrbracket$}{
                     Build prefix-sum array on $\langle W^j_{i, \bu}(\ell): \bu \in R_i\ltimes \bv\rangle$\;
                }
            }
        }
    }
\end{algorithm}

\subsection{Analysis of Optimized Algorithm}
\label{sec:analysis-optimized}
\newcommand{\bw}{\bm{w}}
First, we establish the correctness of Algorithm~\ref{alg:fast-recursive}. 
Specifically, we show that for any call of $\DA(S_\ell, j)$ in Algorithm~\ref{alg:ss-rejected-batch},
Algorithm~\ref{alg:fast-recursive} returns the $\tau$-th tuple in $\mathcal{B}_\ell$, according to a fixed total ordering of the tuples in each bucket $\mathcal{B}_\ell$. Recall from Section~\ref{sec:revisited} that the actual ordering of the tuples is irrelevant, as long as it is fixed. Since Algorithm~\ref{alg:fast-recursive} is deterministic, it suffices to prove that for any two distinct calls of $\DA(S_{\hat{\ell}}, \hat{\tau})$ and $\DA(S_{\bar{\ell}}, \bar{\tau})$ in Algorithm~\ref{alg:ss-rejected-batch},
Algorithm~\ref{alg:fast-recursive} returns two distinct tuples $\hat{\bw}$ and $\bar{\bw}$ from $\join(Q)$.

Formally, we prove the correctness of Algorithm~\ref{alg:fast-recursive} using the following lemma.

\begin{lemma}
Let $\DA(S_{\hat{\ell}}, \hat{\tau})$ and $\DA(S_{\bar{\ell}}, \bar{\tau})$ be two distinct calls in Algorithm~\ref{alg:ss-rejected-batch}.
Let $\hat{\bw}\in \join(Q)$ be the tuple returned by Algorithm~\ref{alg:fast-recursive} on input $(S_{\hat{\ell}}, \hat{\tau})$, and let $\bar{\bw}\in \join(Q)$ be the tuple returned on input $(S_{\bar{\ell}}, \bar{\tau})$. Then $\hat{\bw}\neq \bar{\bw}$.
\end{lemma}

\begin{proof}
We note that due to the correctness of preprocessing, all the $W$ and $M$-values are already correctly computed. 

We first assume that $\hat{\ell}\neq \bar{\ell}$. We show that $\hat{\bw}\in \mathcal{B}_{\hat{\ell}}$ and $\bar{\bw}\in \mathcal{B}_{\bar{\ell}}$. Consider any call of $\DA(S_{\ell'}, \tau')$ in Algorithm~\ref{alg:ss-rejected-batch}. We claim that Algorithm~\ref{alg:fast-recursive} always returns a join result from $\mathcal{B}_{\ell'}$. We prove this by strong induction on the height of the join tree $\T$.
If the join tree has a single leaf node $R_i$ (height $1$), then in Line~1 the algorithm only considers tuples of $R_i$ with score $\ell'$, so the returned tuple belongs to $\mathcal{B}_{\ell'}$.
Assume that the claim holds for all join trees of height at most $h$. Let $R_i$ be the root of a join tree of height $h+1$. In Line~4, the algorithm considers pairs $(\ell_1,\ell_2)$ such that $\ell_1+\ell_2+\phi(\bu)=\ell'$. By the induction hypothesis, we have $\phi(\bu_1)=\ell_1$ and $\phi(\bu_2)=\ell_2$, and therefore $\phi(\bu_1\Join \bu\Join \bu_2)=\ell'$. This proves the claim.

Applying this argument to $\DA(S_{\hat{\ell}},\hat{\tau})$ and $\DA(S_{\bar{\ell}},\bar{\tau})$, we obtain $\hat{\bw}\in \mathcal{B}_{\hat{\ell}}$ and $\bar{\bw}\in \mathcal{B}_{\bar{\ell}}$. Since $\hat{\ell}\neq \bar{\ell}$, we have $\mathcal{B}_{\hat{\ell}}\cap \mathcal{B}_{\bar{\ell}}=\emptyset$, and thus $\hat{\bw}\neq \bar{\bw}$.

We now assume that $\hat{\ell}=\bar{\ell}=\ell^*$ and $\hat{\tau}\neq \bar{\tau}$. Without loss of generality, suppose that $\hat{\tau}>\bar{\tau}$. We prove the claim by strong induction on the height of the join tree.

If the join tree consists of a single leaf node $R_i$ (height $1$), let $\hat{\bu}$ (resp., $\bar{\bu}$) be the smallest tuple found in Line~1 for the pair $(\hat{\ell},\hat{\tau})$ (resp., $(\bar{\ell},\bar{\tau})$). For every $\bu'\in R_i$, the value $W_{i,\bu'}^j(\ell^*)$ is equal to $1$ if $\phi(\bu')=\ell^*$ and $0$ otherwise. Hence, the sum $\sum_{\bu' \in R_i : \bu' \preceq \bu} W^j_{i,\bu'}(\ell^*)$ increases by one for each new qualifying tuple. Since $\hat{\tau}\neq \bar{\tau}$, we conclude that $\hat{\bu}\neq \bar{\bu}$ and therefore $\hat{\bw}\neq \bar{\bw}$.

Assume that the claim holds for all join trees of height at most $h$. Let $R_i$ be the root of a join tree of height $h+1$. If $\hat{\bu}\neq \bar{\bu}$ in Line~1, then clearly $\hat{\bw}\neq \bar{\bw}$. Hence, we assume that $\hat{\bu}=\bar{\bu}=\bu$. The sum $\sum_{\bu' \in R_i : \bu' \preceq \bu} W^j_{i,\bu'}(\ell^*)$ is identical for both instances, so $\hat{\tau}$ and $\bar{\tau}$ are reduced by the same amount in Line~2. By a slight abuse of notation, we again denote the updated values by $\hat{\tau}$ and $\bar{\tau}$.
Let $(\hat{\ell}_1,\hat{\ell}_2)$ and $(\bar{\ell}_1,\bar{\ell}_2)$ be the smallest pairs found in Line~5 for the inputs $(\hat{\ell},\hat{\tau})$ and $(\bar{\ell},\bar{\tau})$, respectively. If $\hat{\ell}_1\neq \bar{\ell}_1$ or $\hat{\ell}_2\neq \bar{\ell}_2$, then by the first part of the proof we immediately obtain $\hat{\bw}\neq \bar{\bw}$. Hence, we assume that $\hat{\ell}_1=\bar{\ell}_1=\ell_1$ and $\hat{\ell}_2=\bar{\ell}_2=\ell_2$.
The sum $\sum_{(\ell'_1,\ell'_2) \in \Phi: (\ell'_1,\ell'_2) \prec (\ell_1, \ell_2)} M_{j, \bu[\key(j)]}(\ell'_1) \cdot W^{\textsf{next}(j)}_{i, \bu}(\ell'_2)$ is identical for both instances, so $\hat{\tau}$ and $\bar{\tau}$ are again reduced by the same amount in Line~6.

We now reach the final stage of the algorithm, where the smallest tuples and score pairs coincide, but $\hat{\tau}>\bar{\tau}$. We compute $\hat{\tau}_1, \hat{\tau}_2$ and $\bar{\tau}_1, \bar{\tau}_2$ as in Lines~7 and~8. The denominator $W_{i,\bu}^{\mathsf{next}(j)}(\ell_2)$ in Line~7 is the same in both executions. If $\hat{\tau}_1\neq \bar{\tau}_1$, then by the induction hypothesis $\hat{\bu}_1\neq \bar{\bu}_1$, which implies $\hat{\bw}\neq \bar{\bw}$.
Otherwise, assume that $\hat{\tau}_1=\bar{\tau}_1=\tau_1$. There are two cases. If neither $\hat{\tau}$ nor $\bar{\tau}$ is divisible by $W_{i,\bu}^{\mathsf{next}(j)}(\ell_2)$, then
$\hat{\tau}_2=((\hat{\tau}-1)\bmod W_{i,\bu}^{\mathsf{next}(j)}(\ell_2))+1=\hat{\tau}\bmod W_{i,\bu}^{\mathsf{next}(j)}(\ell_2)>0$
and
$\bar{\tau}_2=((\bar{\tau}-1)\bmod W_{i,\bu}^{\mathsf{next}(j)}(\ell_2))+1=\bar{\tau}\bmod W_{i,\bu}^{\mathsf{next}(j)}(\ell_2)>0$.
Since $\hat{\tau}>\bar{\tau}$, it follows that $\hat{\tau}_2>\bar{\tau}_2$, and by the induction hypothesis $\hat{\bu}_2\neq \bar{\bu}_2$, again implying $\hat{\bw}\neq \bar{\bw}$.
In the second case, $\hat{\tau}$ is divisible by $W_{i,\bu}^{\mathsf{next}(j)}(\ell_2)$ but $\bar{\tau}$ is not. Then
$\hat{\tau}_2=((\hat{\tau}-1)\bmod W_{i,\bu}^{\mathsf{next}(j)}(\ell_2))+1=W_{i,\bu}^{\mathsf{next}(j)}(\ell_2)$
while
$\bar{\tau}_2=((\bar{\tau}-1)\bmod W_{i,\bu}^{\mathsf{next}(j)}(\ell_2))+1=\bar{\tau}\bmod W_{i,\bu}^{\mathsf{next}(j)}(\ell_2) < W_{i,\bu}^{\mathsf{next}(j)}(\ell_2)$.
Hence $\hat{\tau}_2>\bar{\tau}_2$, and by the induction hypothesis $\hat{\bu}_2\neq \bar{\bu}_2$, which concludes that $\hat{\bw}\neq \bar{\bw}$.
\end{proof}

Next, we bound the preprocessing time, the space, and the query time of our index.

\begin{lemma}
    \label{lem:preprocessing}
    The preprocessing step takes $O(\In \log \In \log\log \In)$ time, and $O(\In \log\In)$ space.
\end{lemma}
\begin{proof}[Proof of Lemma~\ref{lem:preprocessing}]
    For any relation $R_i$, and for every tuple $\bu \in R_i$, we store $O(|C_i|\cdot L) = O(L)$ different $W$ values. So, in total, the space needed for storing all the $W$ values is $O(\In L)$. Moreover, for each $\bv \in R_i[\key(i)]$, we store $O(L)$ different $M$ values. So, in total, the space needed for storing the $M$ values is $O(\In L)$. Hence, the overall space used is $O(\In L) = O(\In \log\In)$. 

    To compute each $W$ value, we use Equation (\ref{W-recursive}). Calculating each $W$ value by (\ref{W-recursive}), needs a summation of $O(L)$ different terms, leading to an $O(NL^2)$ overall running time for calculating all the $W$ values naively. To speed up the process, let's fix a node $R_i \in Q$, a tuple $\bu \in R_i$, and a child node $R_j \in \C_i$ as an example. Now, we need to compute $W^j_{i,\bu}(\cdot)$ for $L$ different scores. There is a convolution structure between $W^{j}_{i, \bu}(\cdot)$, $M_{j, \key(j)}(\cdot)$ and $W^{\textsf{next}(j)}_{i, \bu}(\cdot)$ in Equation~(\ref{W-recursive}), so we can apply the Fast Fourier Transform for computing these $L$ scores together. This can improve the $O(L^2)$ computations to $O(L\log L)$. More formally, consider the multiplication of two degree $L$ polynomials:
    $$p(x) = 
    M_{j,\bu[\key(j)]}(L-1)x^{L-1} + M_{j,\bu[\key(j)]}(L - 2)x^{L - 2} + \cdots + M_{j,\bu[\key(j)]}(0)x^{0}
    ,$$
    and,
    $$q(x) = 
    W^{\textsf{next}(j)}_{i, \bu}(L-1)x^{L-1} + W^{\textsf{next}(j)}_{i, \bu}(L - 2)x^{L - 2} + \cdots + W^{\textsf{next}(j)}_{i, \bu}(0)x^0
    .$$

    By Equation~(\ref{W-recursive}), notice that for all $\ell \in \llbracket L-1 \rrbracket$, the value $W^j_{i,\bu}(\ell)$ is equal to the coefficient of $x^{\ell - \phi(u)}$ in the polynomial $p(x) \cdot q(x)$. Using the Fast Fourier Transform, we can calculate all the coefficients of the degree $2L$ polynomial $p(x)\cdot q(x)$ in $O(L \log L)$ time. Hence, for a fixed node $R_i \in Q$, a tuple $\bu \in R_i$, and a child node $R_j \in \C_i$, we can find all the $L$ different $W^j_{i,\bu}(\cdot)$ values in $O(L \log L)$ time. Putting everything together, calculating all the $W$ values takes $O(NL\log L) = O(\In \log \In \log\log \In)$.
    Having $W$-values, then calculating $M$-values is straightforward using Equation~(\ref{eq:M}) in $O(NL)$ total time. Therefore, the overall running time of the preprocessing step is $O(\In \log \In \log\log \In)$.
   
\end{proof}

\begin{lemma}
\label{lem:bestRun}
Algorithm~\ref{alg:fast-recursive} runs in $O(\log \In)$ time.  
\end{lemma}
\begin{proof}[Proof of Lemma~\ref{lem:bestRun}]
We will analyze its cost step by step. Line~\ref{locate-u} performs a binary search on the prefix-sum tree for $\In$ tuples, which takes $O(\log \In)$ time. Line~\ref{correct-index-u} can also be efficiently done through the prefix-sum tree in $O(\log \In)$ time, since any prefix sum can decomposed into $O(\log \In)$ canonical nodes in this tree. Line \ref{base-case} takes $O(1)$ time. Lines \ref{locate-score-pairs} - \ref{correct-index-score-pairs} take $O(L) = O(\log \In)$ time, since there are at most $L$ pairs in $\Phi$ to be explored. Line~\ref{left-index} takes $O(1)$ time. Hence, all lines before the recursion takes $O(\log \In)$ time. As there are at most $k^2$ different combinations of $(i,j)$, and the recursion is only invoked once for each combination of $(i,j)$, the total number of recursions is $O(k^2)$. As $k$ is a constant, the total running time is $O(\log  \In)$. 
\end{proof}

Finally, we are ready to prove Theorem~\ref{thm:best}.

\begin{proof}[Proof of Theorem~\ref{thm:best}]
The cost of the preprocessing step directly follows from Lemma~\ref{lem:preprocessing}. Below, we focus on the query answering time.
In answering a subset sampling query, whenever some join result needs to be retrieved from $\mathcal{B}_\ell$, we distinguish the following two cases. If $\ell = L$, we evaluate the full join results in $\mathcal{B}_{\ge L}$ (if these join results have not been computed before) to support the direct access. Otherwise, we call the \textsc{\bf RecursiveAccess} procedure accordingly. For the first case, the expected cost is $O(1)$ because there are at most $O(\In^{\rho^*})$ join results in $\mathcal{B}_{\ge L}$ that takes $O(\In^{\rho^*})$ time to materialize~\cite{atserias2013size}, and the probability that at least a join result is retrieved from $\mathcal{B}_{\ge L}$ is at most 
$1-(1-p_L^+)^{|\mathcal{B}_{\ge L}|}\leq 1-(1-\frac{1}{\In^{2\rho^*}})^{\In^{\rho^*}}\leq \frac{1}{\In^{\rho^*}}$, since $|\mathcal{B}_{\ge L}| = \In^{\rho^*}$, and $p^+_L \le \frac{1}{2^{L}} \leq \frac{1}{\In^{2\rho^*}}$.  
Following Lemma~\ref{lm:nss} and Lemma~\ref{lem:bestRun}, the expected query answering time is $O(1+\mu_\Psi\log \In)$.
\end{proof}

\section{Missing Materials in Section~\ref{sec:oneshot}}
\label{appendix:oneshot}

\noindent{\bf Analysis of Algorithm~\ref{alg:one-shotv2}.} 
The correctness follows that of Algorithm~\ref{alg:fast-recursive}, since we exactly simulate the execution of all invocations of Algorithm~\ref{alg:fast-recursive}. 
The data statistics inherited from Theorem~\ref{thm:best} can be computed in $O(\In\log \In\log\log \In)$ time. All $X$-arrays can be computed in $O(\In\log\In)$ time since for every $\bu\in R_i$ there is a unique $\bv\in R_i[\key(i)]$ such that $\bu\in R_i\ltimes \bv$. All $Y$-arrays can be computed in $O(\In\log^2\In)$ time, since there are $O(\In\log\In)$ arrays and each array has size $O(\log\In)$. Using the precomputed $M$ and $W$ statistics, every value in every $Y$ array is computed in $O(1)$ time.

We fix a table $R_i$ and one of its children $R_j$ in $\T$.
Since $|\P_{i,j}|=O(1+\mu_\Psi)$, Radix sort runs in $O(\In+\mu_{\Psi})$ expected time.
To compute the smallest tuple $\bu$ for each quintuple in $\P_{i,j,\bv,\ell}$, we iterate over all quintuples in $\P_{i,j,\bv,\ell}$ and traverse $X_{i,j,\bv,\ell}$ once, in the worst case. Since for every $\bu \in R_i$ there is a unique $\bv$ with $\bu \in R_i\ltimes\bv$ and every quintuple in $\P_{i,j}$ belongs to a unique group $\P_{i,j,\bv,\ell}$, we have $\sum_{\ell}(|X_{i,j,\bv,\ell}|+|\P_{i,j,\bv,\ell}|)=O(\In\log\In+ \mu_\Psi)$.
Similarly, computing all pairs $(\ell_1,\ell_2)$ requires a single pass over all quintuples in each $\P_{i,j,\bu,\ell}$ and all $Y$ tables. Since there are $O(\In\log\In)$ tables $Y$ of size $O(\log\In)$, the total cost is $O(\In\log^2\In + \mu_\Psi)$ expected time. 
The number of nodes in $\T$ is $k=O(1)$ and every node has $O(1)$ children, so, overall, the running time of the one-shot algorithm is $O(\In\log^2\In + \mu_\Psi)$ expected time. 


\begin{algorithm}[t]
    \caption{\textsc{\bf BatchRecursiveAccess}$(i, j, \P_{i,j})$}
    \label{alg:one-shotv2}
    $\mathcal{U} \gets \emptyset$\; 
    Sort $\P_{i,j}$ by $\tau$ using Radix sort\;
    Group $\P_{i,j}$ further by $(i,j,\bv,\ell)$ creating groups $\P_{i,j,\bv,\ell}$ using the sorted order by $\P_{i,j}$\;
    \For{$\ell\in\llbracket L-1\rrbracket$}{
        \For{each $(i,j,\bv,\ell,\tau)\in \P_{i,j,\bv,\ell}$ in sorted order\label{forloop0}}{
            $\bu\gets$ first tuple in $X_{i,j,\bv,\ell}$\;
            \lWhile{$X_{i,j,\bv,\ell}(\bu)<\tau$}{
                $\bu\gets$ next tuple in $X_{i,j,\bv,\ell}$
            }
            \lIf{$R_i$ is a leaf node}{Add $\bu$ as $\vec{\bu}(i,j,\bv,\ell,\tau)$ to $\mathcal{U}$}
            $\P_{i,j}\gets (\P_{i,j}-\{(i,j,\bv,\ell,\tau)\})\cup \{(i,j,\bu,\ell, \tau-X_{i,j,\bv,\ell}(\mathsf{prev}(\bu)))\}$\;
        }   
    }
    \lIf{$R_i$ is a leaf node}{\Return $\mathcal{U}$}
     
    Sort $\P_{i,j}$ by $\tau$ using Radix sort\;
    Group $\P_{i,j}$ further by $(i,j,\bv,\ell)$ creating groups $\P_{i,j,\bv,\ell}$ using the sorted order by $\P_{i,j}$.\;
    $\P_{j,\emptyset}\gets\emptyset, \P_{i,\mathsf{next}(j)}\gets\emptyset$\; 

\For{$\ell\in\llbracket L-1\rrbracket$}{
        \For{each $(i,j,\bu,\ell,\tau)\in \P_{i,j,\bu,\ell}$ in sorted order\label{whileloop1}}{
            $(\ell_1,\ell_2)\gets$ first pair in $Y_{i,j,\bu,\ell}$\;
            \lWhile{$Y_{i,j,\bv,\ell}((\ell_1,\ell_2))<\tau$}{
                $(\ell_1,\ell_2) \gets$ the next pair in $Y_{i,j,\bu,\ell}$
            }
            Compute $\tau_1, \tau_2$ based on $\ell_1, \ell_2$ according to Line \ref{left-index} of Algorithm~\ref{alg:fast-recursive}\;
             $\mathcal{P}_{j,\emptyset} \gets \mathcal{P}_{j,\emptyset} \cup \{(j,\emptyset, \bu,\ell_1,\tau_1)\}$\; 
            $\mathcal{P}_{i,\textrm{next}(j)} \gets \mathcal{P}_{i,\textrm{next}(j)} \cup \{(i,\textrm{next}(j),\bu,\ell_2,\tau_2)\}$\;
        }
        
    }
    $\mathcal{U}\gets \mathcal{U} \cup \textsc{\bf BatchRecursiveAccess}(j, \emptyset, \P_{i,\emptyset})$\;
    $\mathcal{U} \gets \mathcal{U} \cup \textsc{\bf BatchRecursiveAccess}(i, \textrm{next}(j), \P_{i,\textrm{next}(j)})$\;
    \For{$(i,j,\bu,\ell,\tau)\in \P_{i,j}$}{
    $\vec{\bu}(i,j,\bu,\ell,\tau) \gets \vec{\bu}(j,\emptyset, \bu, \ell_1,\tau_1) \Join \bu \Join \vec{\bu}(i, \textrm{next}(j),\bu,\ell_2,\tau_2)$\; \tcp{$\ell_1,\ell_2,\tau_1,\tau_2$ are computed in Lines 17 and 18 above}
    Add $\vec{\bu}(i,j,\bu,\ell,\tau)$ to $\mathcal{U}$\;
    }
    \Return $\mathcal{U}$\;
\end{algorithm}

\section{Discussion on other functions}
\label{appendix:other-functions}

We next discuss the extension to other aggregation functions. For our algorithms, the results can be extended to support $\textsf{MIN}$, $\textsf{MAX}$, and $\textsf{SUM}$. We will show how to adapt the high-level idea in Section~\ref{sec:first-algorithm} to these functions.

\smallskip
\noindent \textbf{$\textsf{MIN}$ and $\textsf{MAX}$.} Without loss of generality, we consider the function $\textsf{MIN}$.
We assume that a \DA oracle is available for the input join $Q$ (under some fixed ordering), such that it receives an integer $i \in [|\join(Q)|]$, and returns the $i$-th element of $\join(Q)$. Similar to before, the ordering can be arbitrary but must remain consistent across multiple invocations of this oracle. 
Let $L=\lceil2\rho^*\log\In\rceil$. 
The first step of partitioning join result is exactly the same as before. 
The second step of applying subset sampling to all sub-instances also follows. Let $\bj = (j_1,j_2,\ldots,j_k)$. Every join result in the sub-instance $Q_{\bj}$ has the probability $\min_{i\in[k]} 2^{-\bj_i} = 2^{-\max_{i\in[k]} j_i}$ to be sampled. We first point out the following observation on each sub-instance in such a partition: 

\begin{lemma}[Either Light or Near-uniform Sub-instance] 
\label{lem:uniform-light-min}
Let $\bj = (j_1,j_2,\ldots,j_k)$. If $\max{i\in[k]} j_i \ge L$, the instance $\Psi_{\bj}$ is light, and otherwise, it is $2$-uniform. 
\end{lemma}

Intuitively, if $\max_{i\in[k]} j_i \ge L$, every join result has its probability at most $\frac{1}{|\join(Q)|}$, so by definition this sub-instance is light. Otherwise, $\max_{i\in[k]} j_i < L$. Now, the ratio between the maximum and the minimum probability is at most $2$, so by definition this sub-instance is $2$-uniform. Hence,
\begin{itemize}[leftmargin=*]
    \item {\bf Preprocessing phase:} We partition input tuples by weights, as well as the join instance $Q$. For each sub-instance $Q_{\bj}$ with $\bj \in \llbracket L\rrbracket^k$, we build \DA oracles and compute the join size.
    \item {\bf Query phase:} We invoke Algorithm~\ref{alg:ss-rejected-batch} with input $\{\langle Q_{\bj}, p\rangle: \bj \in \llbracket L\rrbracket^k\}$ with $p^+_{\bj} = 2^{-\max_{i \in [k]} \bj_i}$. 
    Whenever a join result is sampled from $Q_{\bj}$, we use the \DA oracle to retrieve it.
\end{itemize}

For the optimized index, for every $\bu$ in $R_i$ we define the score $\phi(\bu) = \lfloor - \log \fp_i(\bu) \rfloor$. For a join result $\bu \in \join(Q)$, its score is $\bar{\phi}(\bu) = \min_{i=1}^k \phi(\bu[\schema(R_i)])$. The buckets $\mathcal{B}_\ell=\{\bu\in \join(Q)\mid \bar{\phi}(\bu)=\ell\}$ are defined in a similar way, as for the product function. We note that for any $\bu\in \mathcal{B}_\ell$, $2^{-\ell-1}\leq \fp(\bu)\leq 2^{-\ell}$, hence the ratio between then maximum and minimum probability in each bucket $\mathcal{B}_\ell$ is at most $2$. The index is similar to the one we constructed in Section~\ref{sec:index-optimized}. During the query phase, we run an algorithm similar to Algorithm~\ref{alg:fast-recursive}. One difference is that in steps 4 and 5 all pairs we try are not $L=O(\log \In)$ as we had with the product function but $O(L^2)=O(\log^2 \In)$.
Overall, Theorem~\ref{thm:best} extends to the $\textsf{MIN}$ (and $\textsf{MAX}$) function with the same guarantees up to a $\log\In$ factor.
Similarly, the results for the one-shot and the dynamic variations can also be extended up to an additional $\log\In$ factor.

\smallskip
\noindent \textbf{$\textsf{SUM}$.}
Similarly, our indexing framework naturally extends to the $\textsf{SUM}$ function. Interestingly, $\textsf{SUM}$ can be handled in exactly the same way as $\textsf{MAX}$: we define identical scores and buckets for the $\textsf{SUM}$ function as those used for the $\textsf{MAX}$ function.
Namely,
for every $\bu$ in $R_i$ we define the score $\phi(\bu) = \lfloor - \log \fp_i(\bu) \rfloor$. For a join result $\bu \in \join(Q)$, its score is $\bar{\phi}(\bu) = \max_{i=1}^k \phi(\bu[\schema(R_i)])$. For every $\ell$ we define the bucket $\mathcal{B}_\ell=\{\bu\in \join(Q)\mid \bar{\phi}(\bu)=\ell\}$. We note that for any $\bu\in \mathcal{B}_\ell$, $2^{-\ell-1}\leq \fp(\bu)\leq k\cdot 2^{-\ell}$, hence the ratio between the maximum and minimum probability in each bucket $\mathcal{B}_\ell$ is at most $2k=O(1)$. The index is similar to the one we constructed in Section~\ref{sec:index-optimized}. During the query phase, we run an algorithm similar to Algorithm~\ref{alg:fast-recursive}. One difference is that in steps 4 and 5 all pairs we try are not $L=O(\log \In)$ as we had with the product function but $O(L^2)=O(\log^2 \In)$.
Overall, Theorem~\ref{thm:best} extends to the $\textsf{SUM}$ function with the same guarantees up to a $\log\In$ factor.
Similarly, the results for the one-shot and the dynamic variations can also be extended up to an additional $\log\In$ factor.

\section{Proof of Theorem~\ref{thm:dynamic}}
\label{appndx:dynamic}

\begin{lemma}
    The amortized update time is $O(\log^3 N \log\log \In)$.
\end{lemma}
\begin{proof}
    The key observation is that for each relation $R_i$ and its parent $R_p$, we update the values $\tilde{W}_{p,\bu}^j(\ell')$ if and only if $\tilde{M}_{i,\bv}(\ell)$ changes, where $\bu \in R_p \ltimes \bv$ (note that for each $\bu$ there exists a unique $\bv$ such that $\bu \in R_p \ltimes \bv$).  
    For any fixed $i,\bv,$ and $\ell$, the value $\tilde{M}_{i,\bv}(\ell)$ changes at most $O(\log \In)$ times.

    Whenever an update occurs, we recompute all values $\tilde{W}_{p,\bu}^j(\ell')$ for every $\ell' \in \llbracket L-1 \rrbracket$ using the FFT algorithm, which takes $O(L \log L)$ time. Since there are $L$ possible scores, for fixed $i$ and $\bv$ there are $O(L)$ distinct values $\tilde{M}_{i,\bv}(\ell)$. Hence, the values $\tilde{W}_{p,\bu}^j(\ell')$ are updated at most $O(L \log \In)$ times in total, each costing $O(L \log L)$ time.

    This yields an amortized update time of $O(\log^3 N \log\log \In)$.

    Finally, after every $\In$ insertions, we rebuild the dynamic index from scratch and update $\In \gets 2\In$. This rebuilding step preserves the same amortized update bound of $O(\log^3 N \log\log \In)$.
\end{proof}

\begin{lemma}
    The algorithm correctly updates the values $\tilde{W}$, $\hat{M}$, and $\tilde{M}$.
\end{lemma}
\begin{proof}

    Without loss of generality, assume that all values $\tilde{W}$, $\tilde{M}$, and $\hat{M}$ have been maintained correctly on the subtree $\T_i$ up to the moment when \textsc{Update}$(i, \bv, \ell, \Delta)$ is invoked for the first time at a node $R_i$ located at level $h$ of $\T$. Let $\Delta$ denote the correct increment to $\hat{M}_{i,\bv}(\ell)$.

    We distinguish two cases. If $\tilde{M}_{i,\bv}(\ell)$ does not change, then by Equation~\eqref{eq:tilde-W}, no value $\tilde{W}$ needs to be updated in any ancestor of $R_i$. Since $\tilde{W}$ remains unchanged, all corresponding $\hat{M}$ values on the ancestors of $R_i$ also remain unchanged. Therefore, no further updates are required.

    Otherwise, suppose that $\tilde{M}_{i,\bv}(\ell)$ changes. We then recompute all values $\tilde{W}_{p,\bu}^j(\ell')$ according to Equation~\eqref{eq:tilde-W} (Algorithm~\ref{alg:update}, line~\ref{line:dynW}). If a value $\tilde{W}_{p,\bu}^j(\ell')$ with $j \neq \emptyset$ increases, then no further $\tilde{W}$ or $\hat{M}$ values in the ancestors of $R_i$ need to be updated. If a value $\tilde{W}_{p,\bu}^{\emptyset}(\ell')$ does not increase, again no changes are required in the ancestors.

    Finally, if $\tilde{W}_{p,\bu}^{\emptyset}(\ell')$ increases by a nonzero amount $\Delta'$, then by Equation~\eqref{eq:M-dynamic}, the value $\hat{M}_{i,\bv}(\ell')$ must be increased by $\Delta'$. Consequently, the recursive call $\textbf{Update}(p, \bu[\key(p)], \ell', \Delta')$ is triggered, and $\hat{M}_{i,\bv}(\ell')$ is updated at line~\ref{line:dynM} of Algorithm~\ref{alg:update}. Applying the same argument recursively establishes the correctness of the entire update process.
\end{proof}

\begin{lemma}
    The query procedure runs in $O(1+\mu_\Psi \log N)$ expected time.
\end{lemma}
\begin{proof}
    By definition, there exists a constant $c^\star$, depending only on $|\T|$, such that
\[
W_{i,\bu}^j(\ell) \leq \tilde{W}_{i,\bu}^j(\ell) \leq c^\star W_{i,\bu}^j(\ell)
\quad\text{and}\quad
M_{i,\bu}(\ell) \leq \tilde{M}_{i,\bu}(\ell) \leq c^\star M_{i,\bu}(\ell),
\]
for all $i,j,\ell,$ and $\bu$. 
Therefore, we can simulate the execution of Algorithm~\ref{alg:ss-rejected-batch} (and Algorithm~\ref{alg:fast-recursive}) considering that the number of elements in each bucket $\mathcal{B}_\ell$ is 
$\tilde{M}_{r,\emptyset}(\ell)=O(M_{r,\emptyset}(\ell))$,
where $r$ is the root of $\T$. 
In other words, our index implicitly contains a superset of the true join results that may include dummy tuples.
Since the number of non-dummy tuples is a constant fraction of this superset, we have that the query procedure runs in $O(1+\mu_\Psi \log N)$ expected time. 
Indeed, let $\Psi$ be the instance of subset sampling considering the real $W_{i,\bu}^j(\ell), M_{i,\bu}(\ell)$ and let $\tilde{\Psi}$ be the instance considering the approximations $\tilde{W}_{i,\bu}^j(\ell), \tilde{M}_{i,\bu}(\ell)$. 
For $\ell\in\llbracket L-1\rrbracket$, let $\kappa_\ell$ be the number of tuples from instance $\Psi$ in bucket $\mathcal{B}_\ell$.
For every bucket $\mathcal{B}_\ell$, the instance $\tilde{\Psi}$ contains all tuples from instance $\Psi$ in the same bucket plus $\tilde{c}\cdot \kappa_\ell$ additional dummy tuples, where $\tilde{c}$ is a constant that depends on $|\T|$.
Thus $\mu_{\Psi'}=O(\mu_{\Psi})$ and the result follows.
\end{proof}

\begin{lemma}
    The query procedure returns a valid subset sampling.
\end{lemma}
\begin{proof}
Since $W_{i,\bu}^j(\ell) \leq \tilde{W}_{i,\bu}^j(\ell) \leq c^\star W_{i,\bu}^j(\ell)$ and $M_{i,\bu}(\ell) \leq \tilde{M}_{i,\bu}(\ell) \leq c^\star M_{i,\bu}(\ell),$
for all $i,j,\ell,$ and $\bu$, the correctness follows directly from Algorithm~\ref{alg:ss-rejected-batch} and Lemma~\ref{lm:batched-index}, having $p_\ell^+=c\cdot 2^{-\ell}$, for a constant $c$ that depends on $\T$, and $|S_\ell|=|\mathcal{B}_\ell|=\tilde{M}_{r,\emptyset}(\ell)$, for every $\ell\in \llbracket L-1\rrbracket$. 
\end{proof}


\end{document}